\pdfoutput=1
\documentclass[10pt]{article}
\usepackage[preprint]{tmlr}
\usepackage{microtype}
\usepackage{amsmath,amssymb,amsthm,mathtools}
\usepackage{graphicx,booktabs,multirow,longtable}
\usepackage{algorithm,algpseudocode}
\usepackage{xcolor}
\usepackage{placeins}
\usepackage{float}
\usepackage{caption}
\usepackage[colorlinks=true,linkcolor=blue!50!black,citecolor=blue!50!black,urlcolor=blue!50!black]{hyperref}
\usepackage{url}
\usepackage[capitalize,nameinlink]{cleveref}

\newtheorem{theorem}{Theorem}
\newtheorem{proposition}[theorem]{Proposition}
\newtheorem{corollary}[theorem]{Corollary}

\theoremstyle{definition}

\theoremstyle{remark}
\newtheorem{remark}[theorem]{Remark}
\newtheorem{observation}[theorem]{Observation}

\newcommand{\aeos}{a_{\mathrm{EoS}}}
\newcommand{\Reff}{R_{\mathrm{eff}}}
\newcommand{\R}{\mathbb{R}}
\newcommand{\VRLyapA}{5.32}

\newcommand{\SelfStabSlopeQ}{1.92}
\newcommand{\SelfStabSlopeR}{1.96}
\newcommand{\SelfStabSlopeV}{2.00}
\newcommand{\SelfStabVRFinal}{4.79}
\newcommand{\NetVREight}{$3.7\times10^{-5}$}

\newcommand{\NetQuartEight}{$1.5\times10^{-4}$}
\newcommand{\NetQuartEightC}{0.005}

\newcommand{\QuadParams}{7{,}850}
\newcommand{\QuadLamMax}{300.4}
\newcommand{\QuadErrQuartic}{$1.7\times10^{-9}$}
\newcommand{\QuadErrCubic}{$3.2\times10^{-8}$}
\newcommand{\QuadQuarticOn}{39/39}
\newcommand{\QuadCubicOn}{39/39}
\newcommand{\QuadNoneDiv}{39/39}
\newcommand{\QuadRelOn}{0/24}
\newcommand{\QuadRelPfour}{17}
\newcommand{\QuadVROnStable}{12/15}
\newcommand{\QuadVRAltHigh}{0.11}
\newcommand{\QuadSubBad}{1/18}

\newcommand{\BasinNineHiMin}{0.98}
\newcommand{\BasinNineLoMax}{0.32}
\newcommand{\BasinQuarticNineNine}{0.76}
\newcommand{\BasinQuarticEight}{0.98}
\newcommand{\BasinPQuarterNine}{0.48}
\newcommand{\BasinPQuarterNineNine}{0.01}
\newcommand{\BasinRelNine}{0.27}

\newcommand{\BasinNZero}{300}
\newcommand{\LargeICCubicMin}{0.5}
\newcommand{\LargeICCubicMax}{4.4}

\newcommand{\EMlpRuns}{450}
\newcommand{\EMlpCvxN}{178}
\newcommand{\EMlpCvxMed}{1.006}
\newcommand{\EMlpCvxLo}{1.000}
\newcommand{\EMlpCvxHi}{1.018}
\newcommand{\EMlpCvxWithinTen}{96}

\newcommand{\EMlpVRMed}{0.64}
\newcommand{\EMlpVRLo}{0.34}
\newcommand{\EMlpVRHi}{0.87}
\newcommand{\EMlpGatedAbsDiv}{0/200}
\newcommand{\EMlpMomDivFour}{20/20}
\newcommand{\EMlpMomDivTwo}{1/10}
\newcommand{\EMlpVRSmallRlast}{4.7}
\newcommand{\EMlpVRSmallLossRatio}{316}
\newcommand{\EMlpQSmallLossRatio}{13}
\newcommand{\EMlpVRRelTopMin}{7}
\newcommand{\EMlpVRRelTopMax}{41}

\newcommand{\EMlpLossZero}{0.6}
\newcommand{\EMlpDivHighQuarticRel}{12/20}
\newcommand{\EMlpDivHighCubicRel}{19/20}
\newcommand{\EMlpDivHighRelativisticRel}{6/20}

\newcommand{\EMlpCvxAbsPtsWithinTen}{93}

\newcommand{\EMlpCvxRelPtsWithinTen}{75}

\newcommand{\EMlpQuarticAbsEndWithin}{40/40}

\newcommand{\EMlpQuarticAbsTopLoss}{$3.1{\times}10^{-6}$}
\newcommand{\EMlpCubicAbsReffMed}{1.003}
\newcommand{\EMlpCubicAbsReffLo}{1.000}
\newcommand{\EMlpCubicAbsReffHi}{1.006}
\newcommand{\EMlpCubicAbsEndWithin}{0/40}
\newcommand{\EMlpCubicAbsEndLo}{1.01}
\newcommand{\EMlpCubicAbsEndHi}{3.32}
\newcommand{\EMlpCubicAbsTopLoss}{0.20}

\newcommand{\EMlpRelativisticAbsEndWithin}{40/40}

\newcommand{\EMlpRelativisticAbsTopLoss}{$1.9{\times}10^{-7}$}

\newcommand{\EThKinVsMuonAcc}{+0.17}
\newcommand{\EThKinVsMuonAccSd}{0.13}
\newcommand{\EThKinVsMuonAccP}{0.04}

\newcommand{\EThKinVsMuonCE}{-0.016}
\newcommand{\EThKinVsMuonCESd}{0.011}
\newcommand{\EThKinVsMuonCEP}{0.03}

\newcommand{\EThKinVsEffAcc}{+0.10}
\newcommand{\EThKinVsEffAccSd}{0.20}
\newcommand{\EThKinVsEffAccP}{0.33}
\newcommand{\EThKinVsEffAccLo}{-0.15}
\newcommand{\EThKinVsEffAccHi}{+0.34}
\newcommand{\EThKinVsEffCE}{-0.009}
\newcommand{\EThKinVsEffCESd}{0.015}
\newcommand{\EThKinVsEffCEP}{0.28}
\newcommand{\EThKinVsEffCELo}{-0.028}
\newcommand{\EThKinVsEffCEHi}{+0.010}

\newcommand{\EThEffVsMuonCE}{-0.008}

\newcommand{\EThEffVsMuonCEP}{0.07}

\newcommand{\EThKqVsThreeAcc}{+0.05}

\newcommand{\EThKqVsThreeCE}{+0.008}

\newcommand{\EThKqVsThreeMinP}{0.17}
\newcommand{\EThKinGate}{0.70}

\newcommand{\EThSeeds}{5}

\newcommand{\DelayPredErr}{24}
\newcommand{\DelayPredErrSlow}{3}

\newcommand{\DelayMlpPeakObs}{1.00}
\newcommand{\DelayMlpReffCross}{0.55}
\newcommand{\DelayMlpRmax}{1.75}

\newcommand{\DelayResN}{15}
\newcommand{\DelayResSpearman}{0.82}
\newcommand{\DelayResSpearmanP}{2\times10^{-4}}
\newcommand{\DelayResFastN}{9}
\newcommand{\DelayResSlowN}{6}
\newcommand{\DelayResFastLo}{2.0}
\newcommand{\DelayResFastHi}{3.6}
\newcommand{\DelayResSlowHi}{1.17}
\newcommand{\DelayResAboveBand}{11}

\newcommand{\DelayMlpMomPeakHi}{1.06}
\newcommand{\DelayMlpMomRhoHi}{0.010}

\newcommand{\EResRuns}{135}
\newcommand{\EResCvxN}{37}
\newcommand{\EResCvxMed}{1.363}
\newcommand{\EResCvxLo}{1.215}
\newcommand{\EResCvxHi}{1.578}

\newcommand{\EResAbsHighN}{40}

\newcommand{\EResRelHighN}{40}
\newcommand{\EResRelHighTrained}{2}
\newcommand{\EResRelHighCollapsed}{28}
\newcommand{\EResRelHighDiv}{9}
\newcommand{\EResMomHighN}{10}

\newcommand{\EResMomHighCollapsed}{5}
\newcommand{\EResMomHighDiv}{5}
\newcommand{\EResAbsHighMaxLoss}{$2.7{\times}10^{-4}$}
\newcommand{\EResMomCollLo}{0.060}
\newcommand{\EResMomCollHi}{0.090}
\newcommand{\EResMomCollRminLo}{0.010}
\newcommand{\EResMomCollRminHi}{0.10}

\newcommand{\ETwoRuns}{24}

\newcommand{\ETwoRlateLo}{0.91}
\newcommand{\ETwoRlateHi}{1.00}

\newcommand{\ETwoMgateLateMin}{0.99}
\newcommand{\ETwoVrLateMin}{1.00}
\newcommand{\ETwoQgateLateMin}{0.61}
\newcommand{\ETwoMgateEarlyMax}{0.15}
\newcommand{\ETwoVrEarlyMin}{0.10}
\newcommand{\ETwoVrEarlyMax}{0.60}
\newcommand{\ETwoMgateEarlyMin}{0.12}
\newcommand{\ETwoQgateEarlyMin}{0.45}
\newcommand{\EFourRuns}{16}
\newcommand{\EFourRLo}{0.16}
\newcommand{\EFourRHi}{0.27}
\newcommand{\EFourLegLo}{0.87}
\newcommand{\EFourLegHi}{1.00}
\newcommand{\EFourBsLo}{0.073}
\newcommand{\EFourBsHi}{0.112}
\newcommand{\EFourEtaBsLo}{2.8}
\newcommand{\EFourEtaBsHi}{4.3}
\newcommand{\EFourLamLo}{-0.11}
\newcommand{\EFourLamHi}{-0.06}

\newcommand{\EFiveAdamwBest}{3.714}

\newcommand{\EFiveClipBest}{3.716}

\newcommand{\EFiveVrBest}{3.705}

\newcommand{\EFiveKinOneBest}{3.737}

\newcommand{\EFiveKinFourBest}{3.706}
\newcommand{\EFiveAdamwTen}{6.47 and 6.69}
\newcommand{\EFiveClipTen}{3.79 and 4.07}
\newcommand{\EFiveVrTen}{5.99 and 6.09}
\newcommand{\EFiveVrGateThree}{1.00}
\newcommand{\EFiveVrGateTen}{0.19--0.25}
\newcommand{\EFiveKinOneTen}{3.96 and 4.84}
\newcommand{\EFiveKinOneGateThree}{0.54}

\newcommand{\EFiveKinFourGateThree}{0.78}

\newcommand{\EFiveAdamwSeedSpread}{0.026}
\newcommand{\EFiveBestSpread}{0.032}
\newcommand{\EFiveParams}{51.5}
\newcommand{\EFiveTokens}{393}
\newcommand{\EFiveMinutes}{16}
\newcommand{\EFiveRuns}{40}

\newcommand{\ESixWdOneN}{15}
\newcommand{\ESixGrokLo}{11.4}
\newcommand{\ESixGrokHi}{21.9}
\newcommand{\ESixSpearman}{0.22}
\newcommand{\ESixSpearmanP}{0.43}
\newcommand{\ESixSpearmanPre}{0.50}
\newcommand{\ESixSpearmanPreP}{0.06}
\newcommand{\ESixPreN}{15}

\newcommand{\ESixWdZeroAccHi}{0.27}
\newcommand{\ESixAdamwSpikeOne}{0.57}
\newcommand{\ESixAdamwGrok}{14.6}

\newcommand{\ESixClipSpikeOne}{0.14}
\newcommand{\ESixClipGrok}{18.4}

\newcommand{\ESixKinOneSpikeOne}{0.13}
\newcommand{\ESixKinOneGrok}{15.6}

\newcommand{\RampKickRuns}{84}
\newcommand{\RampKickValid}{81}
\newcommand{\RampKickExcluded}{3}

\newcommand{\RampKickSeriesErr}{0.041}
\newcommand{\RampKickSeriesErrMax}{25}
\newcommand{\RampKickSeriesErrMaxFast}{2}
\newcommand{\RampOrdersDown}{8}

\newcommand{\RampExitDeep}{1.115}
\newcommand{\RampExitShallow}{1.020}
\newcommand{\RampRinDeep}{0.1}
\newcommand{\RampRinShallow}{0.9}
\newcommand{\RampPsiDeep}{1.118}
\newcommand{\RampPsiShallow}{1.029}

\newcommand{\RampPsiDiffThree}{0.010}
\newcommand{\RampPsiDiffOne}{0.004}

\newcommand{\RampIdealDiff}{0.009}

\newcommand{\RampLamLo}{206}
\newcommand{\RampLamHi}{211}
\newcommand{\RampLossLo}{0.015}
\newcommand{\RampLossHi}{0.017}
\newcommand{\RampLamDrift}{9.5}
\newcommand{\RampNoiseRuns}{45}
\newcommand{\RampNoiseInBand}{43}

\newcommand{\RampKnRuns}{45}
\newcommand{\RampKnInBand}{42}

\newcommand{\RampKnLoDeep}{1.090}
\newcommand{\RampKnLoShallow}{1.020}

\newcommand{\RampKnMidDeepLo}{1.05}
\newcommand{\RampKnMidDeepHi}{1.06}

\newcommand{\RampLamBelow}{-0.0527}
\newcommand{\RampLamA}{0.15}
\newcommand{\RampLamB}{0.57}

\newcommand{\ESensCvxLo}{1.004}
\newcommand{\ESensCvxHi}{1.009}
\newcommand{\ESensVrLo}{0.62}
\newcommand{\ESensVrHi}{0.71}
\newcommand{\ESensCvxNLo}{114}
\newcommand{\ESensCvxNHi}{213}

\newcommand{\RampTestErrA}{25}
\newcommand{\RampTestErrB}{8.7}
\newcommand{\RampTestErrC}{1.9}

\title{Physics-Inspired Optimizers beyond the Edge of Stability:\\
Period-Two Orbits and Delayed Instability}
\author{\name Alexander Itin\thanks{E-mail: \texttt{alx.itin2@gmail.com}; website: \href{https://ai4science.de}{AI4science.de}} \\
\addr Radboud University, Nijmegen, The Netherlands \\
\addr Independent researcher, Hamburg, Germany}

\begin{document}
\maketitle
\begingroup\renewcommand{\thefootnote}{}\footnotetext{The author used large language model assistants (Claude Opus 5.5, Anthropic; GPT-6 Astra Pro, OpenAI) for programming, running and analysing numerical experiments, reviewing drafts, and editing the text. The research questions, the theoretical framework and the interpretation of the results are the author's own; the author has checked the derivations, code and reported numbers and takes full responsibility for the content.}\endgroup

\begin{abstract}
Gradient descent with momentum can be viewed as a particle that moves with friction in the loss landscape. In deep learning the learning rate is usually large, so training is a discrete dynamical system, not a smooth flow. In many full-batch experiments the sharpness of the loss grows to the stability limit of the optimizer, the edge of stability (EoS), and training continues there, oscillating along the sharpest direction. Physics-inspired optimizers change this motion: Velocity-Regularized Adam (VRAdam), for example, shrinks the step at high speed, motivated by a quartic kinetic energy. We study such methods as \emph{kinetic maps}, in which the step is the gradient of a kinetic energy of the momentum, and ask what they do beyond the EoS. Our main findings are as follows. (i)~In a reduced model, a kinetic map creates a symmetric period-two orbit beyond the edge, where plain momentum diverges. This orbit is strictly linearly stable if and only if the elasticity of the map at its amplitude (the log-log slope of step length versus momentum) lies between 0 and 1; on it, the step reduction exactly compensates the excess sharpness, giving an effective edge of stability. VRAdam's damper fails this test above twice the critical learning rate, where its kinetic energy is not convex. The Legendre-dual map of VRAdam's own Lagrangian is convex and passes the test at every learning rate. (ii)~Our central dynamical finding is that, because the sharpness rises gradually, the loss of stability is delayed, as in Neishtadt's theory of slowly swept bifurcations. In a network with a controlled sweep, the measured dynamics match this theory without fitted parameters, including the memory of where the sweep started and its erasure by noise. (iii)~Local stability does not guarantee that training reaches the orbit. With momentum, a trajectory can instead end on another attractor, such as a period-four cycle or chaos, unless the map keeps enough elasticity at large momentum; the amount needed grows with the momentum coefficient $\beta_1$. Maps that bound the step, such as the relativistic one, already fail at $\beta_1=0.9$. The cubic-Legendre map (elasticity tending to $1/2$) reached the orbit from every sampled start up to $\beta_1=0.99$, at the price of an oscillation that grows faster with the learning rate. Experiments support the theory on an exact quadratic problem and in a full-batch MLP, where convex maps hold training near the effective edge while VRAdam's damper gets stuck at a strongly reduced step. In ResNet-32, maps that act early prevent collapse at large learning rates. On ResNet-18, however, much of the apparent benefit of the maps is reproduced by matching the effective learning rate. Code and data will be made publicly available.
\end{abstract}

\section{Introduction}
Gradient descent is one of the oldest algorithms of applied mathematics: Cauchy proposed it in 1847 for astronomical computations \citep{cauchy1847}. A century later Polyak added momentum, the \emph{heavy ball}, by viewing the iteration as a particle that moves with friction in the potential given by the loss \citep{polyak1964heavyball}, and Nesterov's acceleration followed \citep{nesterov1983}. This mechanical picture has since grown into a theory of its own. For small steps, accelerated methods are discretizations of damped second-order ODEs \citep{su2016ode} and of Bregman--Lagrangian flows \citep{wibisono2016variational}, and symplectic integrators discretize such flows \citep{betancourt2018symplectic}. Hamiltonian descent chooses the kinetic energy $K$ to suit the loss \citep{maddison2018hamiltonian}, relativistic kinetic energies bound the speed \citep{franca2020conformal}, and Lion, Adam and heavy ball can be written as Hamiltonian dynamics with descent and a general kinetic energy \citep{chen2024lionk,liang2026cautious}. Stochastic gradients add a picture based on modified equations and diffusion \citep{li2017sme}.

Deep networks, however, are trained with large steps, and at practical learning rates training is a discrete dynamical system far from its continuous limit. Gradient descent with large constant steps can show period doubling and chaos \citep{chen2024stability,kong2020stochasticity}, the route to chaos of the logistic map \citep{feigenbaum1978}, and some variable-step variants are chaotic \citep{vandendoel2012chaotic}. Large learning rates make networks \emph{catapult} out of sharp regions \citep{lewkowycz2020catapult}, and the learning rate in the first phase of training determines the curvature and gradient noise later on \citep{jastrzebski2020breakeven}. \citet{cohen2021eos} identified a regime in which full-batch training often operates: the sharpness (the largest eigenvalue of the Hessian) rises until it reaches the stability threshold of the optimizer, $2/\eta$ for gradient descent, and then \emph{stays there}, while the iterates oscillate along the top eigenvector and the loss keeps decreasing on average. This \emph{edge of stability} (EoS) also appears for adaptive methods \citep{cohen2022adaptive}, is maintained by self-stabilization \citep{damian2023selfstab} and has a stochastic counterpart \citep{andreyev2024eoss}. In dynamical-systems terms, the minimum has lost stability through a period-doubling bifurcation, and training proceeds near the period-two orbit born there. We therefore focus on how an optimizer changes this oscillatory regime after the minimum loses stability.

Velocity-Regularized Adam (VRAdam) \citep{vaidhyanathan2026vradam} is a recent physics-inspired optimizer. Adding a quartic term $\beta_3 v^4/4$ to the kinetic energy of a particle makes its effective mass grow with speed, which suggests damping the learning rate as
\begin{equation}\label{eq:vradam}
\eta_t=\frac{\alpha_0}{1+\min(\beta_3\|m_t\|^2,\alpha_1)},
\end{equation}
where $m_t$ is Adam's first moment (written $v_t$, a velocity, in VRAdam) and the norm is taken over all parameters. The analysis of VRAdam covers the regime \emph{below} the EoS threshold, where the damper does not change the linear stability of the minimum. We ask what such kinetic modifications do \emph{above} the threshold. Our contributions are:
\begin{enumerate}
\item \textbf{Period-two orbits, local stability and global attraction (\cref{sec:kinetic,sec:theory,sec:global}).} Following the terminology of Hamiltonian descent \citep{maddison2018hamiltonian}, we call updates $u=\nabla K(q)$ with a radial kinetic energy $K$ \emph{kinetic maps}; our results concern their discrete dynamics beyond the stability threshold. With momentum, the symmetric period-two orbit beyond the edge is strictly linearly stable if and only if the elasticity of the map at the orbit lies in $(0,1)$ (\cref{thm:main}), and on the orbit the effective step is exactly at the edge, $\eta\lambda_{\max}g/a_{\mathrm{EoS}}=1$ (\cref{cor:gated}). VRAdam's damper is a non-convex member of the family, and its orbit loses stability at twice the threshold step; the Legendre dual of its own Lagrangian is convex and keeps the orbit stable at every step. Whether trajectories reach the orbit depends on how much elasticity the map keeps at large momentum, increasingly so as $\beta_1$ grows; an exact quadratic problem built from real data reproduces the closed forms to a relative error below $10^{-7}$.
\item \textbf{Delayed passage through the edge (\cref{sec:delay}).} Because the sharpness rises gradually, the bifurcation is crossed slowly, and the loss of stability is delayed as described by Neishtadt's theory of dynamical bifurcations. In a network with a controlled sweep, the deviation along the top Hessian direction follows the linear theory without fitted parameters; for slow sweeps the exit point obeys Neishtadt's entry--exit relation, and noise erases this memory as predicted.
\item \textbf{Neural-network dynamics and optimizer benchmarks (\cref{sec:networks,sec:practice}).} In a full-batch MLP the effective-edge law holds and the VR map fails as predicted; in ResNet-32 the orbit forms too late, but maps that act early prevent catapult collapse; with Adam, maps on the preconditioned direction cannot see the amplitude of the oscillation. On ResNet-18, controls with matched effective learning rates remove most of the difference for an Adam map and about half of the gain of Kinetic Muon; the remainder is not resolved by five seeds.
\end{enumerate}
\Cref{sec:discussion} discusses connections to capture into resonance, self-organized criticality, ponderomotive forces and feedback control, reports a negative result on grokking, and summarizes limitations and conclusions.

\section{Background and kinetic maps}\label{sec:related}
\subsection{Related work}
\paragraph{Edge of stability.} For gradient descent with step $\eta$ on a quadratic, the minimum is stable if and only if $\eta\lambda_{\max}<2$. With EMA momentum $\beta_1$ the threshold becomes $\eta\lambda_{\max}<2(1+\beta_1)/(1-\beta_1)$, and for Adam the relevant curvature is the preconditioned sharpness $\lambda_{\max}(P^{-1}H)$ with $P=\mathrm{diag}(\sqrt{\hat s}+\epsilon)$, which gives $38/\eta$ at $\beta_1=0.9$ \citep{cohen2021eos,cohen2022adaptive}. Self-stabilization explains how the oscillation along the top eigenvector reduces the sharpness through third-order terms \citep{damian2023selfstab}; normalized gradient descent was analysed by \citet{arora2022eos}. Period-two orbits beyond the threshold have been studied in low-dimensional models \citep{ma2022beyond,chen2023beyond,zhu2023minimalist,song2023trajectory,kong2020stochasticity} and, for gradient descent, derived from a variational functional, with a supercritical bifurcation in two-layer linear networks \citep{litman2026origin}; \citet{wang2025regularity} argue, with proofs for a class of nonconvex functions and experiments on networks, that good regularity of the objective makes edge of stability, balancing and catapults more likely at large learning rates. \citet{ma2022beyond} prove that one-dimensional gradient descent converges to the period-two orbit from almost every start for even losses whose curvature decreases away from the minimum. This global result is stronger than a local analysis; without momentum, our criterion concerns only the linear stability of the orbit. Clipped gradient descent has strictly clipped two-cycles with a unit multiplier \citep{yuanwen2026clipped}, and \citet{fong2026adam} analyse Adam's own oscillations. In minibatch training the full-batch sharpness stays below the threshold while a \emph{batch sharpness} sits at it \citep{andreyev2024eoss}; momentum changes this plateau in a way that depends on the batch size \citep{andreyev2026momentum}. For scale-invariant (normalized) networks, \citet{kodryan2022three} distinguish convergence, chaotic-equilibrium and divergence regimes as the effective learning rate grows, \citet{lobacheva2021periodic} describe periodic destabilizations under batch normalization and weight decay, and \citet{sadrtdinov2024where} find that the best initial learning rates lie just above the convergence threshold.

\paragraph{Bifurcations in training.} Bifurcations also occur in the dynamics of the trained model itself: in ReLU-based recurrent networks, certain bifurcations of the network's state-space dynamics crossed during training come with diverging or vanishing gradients and can explain sudden loss jumps \citep{eisenmann2023bifurcations}. The bifurcations studied here are of a different kind: they belong to the optimizer map, with the learning rate and the sharpness as parameters.

\paragraph{Kinetic energies in optimization.} Hamiltonian descent \citep{maddison2018hamiltonian} matches a kinetic energy to the growth of the objective; relativistic kinetic energies bound the speed \citep{franca2020conformal}; Lion-$\mathcal K$ \citep{chen2024lionk} uses a general convex $\mathcal K$, and Cautious optimizers \citep{liang2026cautious} a general kinetic energy; \citet{karoni2026nonlinear} add damping that depends on the kinetic energy. New here are the dynamics beyond the EoS, the elasticity laws, and the identification of VRAdam's damper as a non-convex member of the family.

\paragraph{Update damping and learning-rate control.} Clipping in various forms suppresses loss spikes \citep{shazeer2018adafactor,wortsman2023stable,wang2025adagc,kumar2025zclip,huang2025spam}, and curvature-aware methods set the step from the estimated curvature \citep{roulet2024stepping,malitsky2020adgd}. We include a clipping baseline and two curvature controllers.

\subsection{Kinetic maps}\label{sec:kinetic}
In Hamiltonian descent, the gradient of a kinetic energy is called a kinetic map \citep{maddison2018hamiltonian}; we consider radial ones (notation is collected in \cref{tab:notation}). Let $q$ be the quantity that an optimizer turns into a step (heavy ball: the momentum; Adam: $\hat m/(\sqrt{\hat s}+\epsilon)$, with $\hat s$ the second moment). A \emph{kinetic map} with scale $c>0$, acting on a block $B$ of parameters (a tensor, or all parameters), is
\begin{equation}\label{eq:map}
u_B=\psi(q_B)=q_B\,g\!\left(\tfrac{\|q_B\|}{c}\right),\qquad g(0)=1,
\end{equation}
the gradient of the radial kinetic energy $K(q)=k(\|q\|)$ with $k'(r)=r\,g(r/c)$. We call $g$ the \emph{step multiplier} and $\phi(z)=z\,g(z)$ the radial \emph{speed profile}, in units of $c$. $K$ is convex if and only if $\phi$ is non-decreasing. The \emph{elasticity} of the map is
\begin{equation}\label{eq:elasticity}
\varepsilon(z)=\frac{z\,\phi'(z)}{\phi(z)}=\frac{d\ln\phi}{d\ln z}.
\end{equation}
This is the standard elasticity of a function, familiar from economics: the relative change of the speed per relative change of the momentum. A map with $\varepsilon\equiv1$ does not damp; $0<\varepsilon<1$ means that the speed grows more slowly than the momentum; $\varepsilon<0$ (a non-convex $K$) means that the speed falls as the momentum grows. Five maps matter here (\cref{fig:maps}a,b):
\begin{itemize}
\item \textbf{VR map} $\phi(z)=z/(1+z^2)$: VRAdam's damper with $c=\beta_3^{-1/2}$ (ignoring its floor). $K=\tfrac{c^2}{2}\log(1+z^2)$ is \emph{non-convex} for $z>1$, where the speed decreases with momentum.
\item \textbf{Quartic-Legendre map} $\phi=s$ with $s+s^3=z$: the velocity of VRAdam's own Lagrangian $L=\tfrac12v^2+\tfrac14v^4$, obtained by inverting $p=\partial L/\partial v$. Convex; $\varepsilon\to1/3$.
\item \textbf{Cubic-Legendre map} $\phi=s$ with $s+s^2=z$, i.e.\ $s=2z/(1+\sqrt{1+4z})$ and $g=2/(1+\sqrt{1+4z})$, from $L=\tfrac12v^2+\tfrac13|v|^3$. Convex; $\varepsilon\to1/2$.
\item \textbf{Relativistic map} $\phi(z)=z/\sqrt{1+z^2}$, $K=c^2\sqrt{1+z^2}$ \citep{franca2020conformal}: convex, speed bounded by $c$, $\varepsilon\to0$.
\item \textbf{Norm clipping} $\phi(z)=\min(z,1)$: Huber $K$, convex but not strictly; $\varepsilon=0$ when clipped.
\end{itemize}
VRAdam derives the equation of motion $\dot v=-\nabla f/(m+3\beta_3v^2)$ correctly, but then uses $1/(1+\beta_3v^2)$ as a multiplier of the velocity in the position update. The direct damping rule and the Legendre-dual construction therefore differ in convexity.

\begin{figure}[htbp]
\centering
\includegraphics[width=\linewidth]{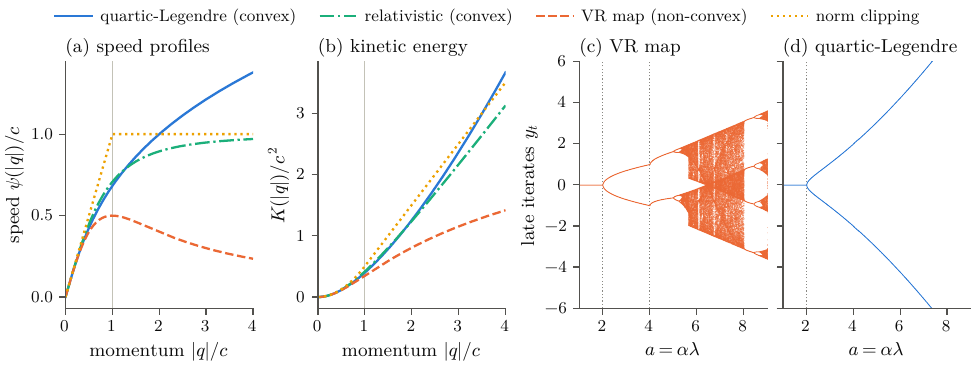}
\caption{(a) Radial speed profiles and (b) the kinetic energies they imply. (c,d) Late iterates of the reduced model without momentum, $y'=y-a\phi(y)$. The VR map has a stable symmetric two-cycle only for $2<a<4$, followed by a symmetry-breaking pitchfork, period doubling and chaos (first positive Lyapunov exponent at $a\approx\VRLyapA$); the quartic map has a stable two-cycle for every $a>2$.}
\label{fig:maps}
\end{figure}

\section{Period-two orbits and local stability}\label{sec:theory}
\paragraph{Reduced model.} Consider one eigen-direction of a quadratic with curvature $\lambda$ (for Adam: an eigenvalue of $P^{-1/2}HP^{-1/2}$ with the preconditioner frozen), EMA momentum $b=\beta_1\in[0,1)$, and a kinetic map acting on the momentum. In units of $c$, with $y$ the gradient and $v$ the momentum,
\begin{equation}\label{eq:reduced}
v'=b\,v+(1-b)\,y,\qquad y'=y-a\,\phi(v'),\qquad a=\eta\lambda .
\end{equation}
We assume that $\phi$ is odd and $C^1$ with $\phi'(0)=1$, $\phi>0$ on $(0,\infty)$, and that $g(z)=\phi(z)/z$ decreases strictly to a limit $g_\infty\in[0,1)$. Let $\kappa=(1-b)/(1+b)$, $A=\kappa a$ and $\aeos=2/\kappa=2(1+b)/(1-b)$.

\begin{theorem}[EoS orbits of kinetic maps]\label{thm:main}
(i) The minimum is linearly stable if and only if $A<2$, i.e.\ $a<\aeos$.
(ii) For $2<A<2/g_\infty$ there is exactly one symmetric period-two orbit $(y,v)\mapsto(-y,-v)$; on it $v'=\pm z^\ast$ with $g(z^\ast)=2/A$ and $|y|=z^\ast/\kappa$.
(iii) The one-step Jacobian at the orbit has the characteristic polynomial $\nu^2-(1+b)(1-2\varepsilon^\ast)\nu+b$ with $\varepsilon^\ast=\varepsilon(z^\ast)$, and the two-step Jacobian is its square. The orbit is strictly linearly stable (spectral radius below one) if $0<\varepsilon^\ast<1$ and unstable if $\varepsilon^\ast<0$. At $\varepsilon^\ast=0$ it is nonhyperbolic: the two-step Jacobian has a unit multiplier.
\end{theorem}
\begin{proof}[Proof sketch] (ii) $(y,v)\mapsto(-y,-v)$ requires $v=-\kappa y$ and $v'=\kappa y=:z$; then $y'=-y$ gives $a\phi(z)=2z/\kappa$, i.e.\ $g(z)=2/A$. (iii) With $d=\phi'(z^\ast)$, $J=\begin{psmallmatrix}1-ad(1-b)&-adb\\1-b&b\end{psmallmatrix}$, $\det J=b$, $\operatorname{tr}J=(1+b)(1-Ad)$ and $Ad=2\varepsilon^\ast$. $J$ is the same at both points of the orbit because $\phi'$ is even, and the Jury criterion gives $|1-2\varepsilon^\ast|<1$. Details are in \cref{app:proofs}.
\end{proof}
The VR map's orbit loses stability at $\varepsilon^\ast=0$ ($A=4$) through a symmetry-breaking pitchfork, which can be described exactly (\cref{prop:vrasym} in \cref{app:proofs}): for every $b$ a pair of mirror-image asymmetric two-cycles branches off at $A=4$, and it is strictly linearly stable for $4<A<4+(1+b^2)/(1+b)^2$. Without momentum this is $4<A<5$, and the asymmetric cycles period-double exactly at $A=5$ (\cref{fig:maps}c).

\begin{corollary}[Effective edge of stability]\label{cor:gated}
On the orbit the step multiplier is $g(z^\ast)=2/A=\aeos/a$, i.e.
\begin{equation}\label{eq:gated}
\Reff\;:=\;\frac{\eta\,\lambda\,g(z^\ast)}{\aeos}\;=\;1 .
\end{equation}
The curvature can exceed the plain EoS threshold by the factor $1/g$; the oscillation amplitude adjusts so that the \emph{effective} step sits exactly at the edge. Directions with smaller positive curvature see the same multiplier and contract (\cref{prop:blocks}).
\end{corollary}

\begin{corollary}[Closed forms]\label{cor:closed}
For $A>2$ (amplitude $z^\ast$, elasticity $\varepsilon^\ast$, stability):
\begin{center}\small
\begin{tabular}{lccc}
\toprule
map & amplitude & $\varepsilon^\ast$ & locally stable for \\
\midrule
VR & $z^{\ast2}=A/2-1$ & $(4-A)/A$ & $2<A<4$ \\
relativistic & $z^{\ast2}=(A/2)^2-1$ & $4/A^2$ & all $A>2$ \\
quartic-Legendre & $s^{\ast2}=A/2-1,\ z^\ast=s^\ast A/2$ & $A/(3A-4)\to1/3$ & all $A>2$ \\
cubic-Legendre & $s^\ast=A/2-1,\ z^\ast=s^\ast(1+s^\ast)$ & $A/(2A-2)\to1/2$ & all $A>2$ \\
clipping$^\dagger$ & continuum of orbits & $0$ & neutral \\
\bottomrule
\end{tabular}
\end{center}
VRAdam's damper keeps a stable EoS orbit only up to twice the threshold step, for every $\beta_1$. $^\dagger$Clipping does not satisfy the smoothness and strict-decrease assumptions of \cref{thm:main}; it is listed as a limiting case.
\end{corollary}

\begin{remark}[Power laws and VRAdam's floor]
If $\varepsilon(z)\to p$ (e.g.\ $\phi\sim z^p$), then $\varepsilon^\ast\to p$ as $A\to\infty$; a kinetic term $|v|^{n+1}/(n+1)$ gives $p=1/n$. VRAdam's floor $1/(1+\alpha_1)$ makes the map linear for $|v'|>\sqrt{\alpha_1}$; with $\alpha_1=19$ the relevant regime is $2\aeos<a\lesssim20\aeos$, where the VR map's symmetric orbit is unstable and the reduced model shows asymmetric cycles, period doubling and chaos.
\end{remark}

The multiplier acts on every direction of the block on which it is computed, which matters for how the map is applied:
\begin{proposition}[Global versus per-block maps]\label{prop:blocks}
Let $f=\tfrac12\sum_i\lambda_i x_i^2$, $b=0$, and one map on the norm of the full gradient. If the top direction is on the stable orbit with $\eta\lambda_{\max}>2$ and the other coordinates are small, each other coordinate evolves as $x_i'=(1-2\lambda_i/\lambda_{\max})x_i$: the map slows every other direction down to the rate of gradient descent at the edge. With one map per block, blocks that do not contain the top direction keep the step $\eta$ to linear order (their own multipliers are $1+O(\|q_B\|/c)$ for the cubic map and $1+O(\|q_B\|^2/c^2)$ for the quartic, relativistic and VR maps).
\end{proposition}

Without momentum, the VR map has a positive Lyapunov exponent from $A\approx\VRLyapA$ (\cref{fig:maps}c); with $b=0.9$ the Floquet radius of its symmetric orbit exceeds one for $a/\aeos>2$, as \cref{thm:main} gives, while for the convex maps it stays below one (\cref{fig:momentum} in \cref{app:reduced}).

\section{Delayed passage through the edge}\label{sec:delay}
In training, the step is usually fixed while the sharpness rises (progressive sharpening). The bifurcation parameter $R=\eta\lambda_{\max}/\aeos$ therefore moves slowly through $R=1$ instead of being fixed. Bifurcations crossed in this way are called \emph{dynamical bifurcations}. The central result, due to Neishtadt \citep{neishtadt1987delay,neishtadt1988delay} and reviewed by \citet{neishtadt2021stability}, is that when a parameter slowly crosses a stability threshold, the system stays near the equilibrium, now unstable, for a long time, of order $1/\rho$ steps for a sweep at rate $\rho$, before it leaves. We state the linear theory for our map and then test it in a network.

\paragraph{Linear theory.} Near the minimum, a small deviation $w_t=(y_t,v_t)$ of \eqref{eq:reduced} obeys $w_{t+1}=J_0(R_t)\,w_t$, where $J_0$ is the Jacobian at the minimum (\cref{app:proofs}). Let $\Lambda(R)=\log r(J_0(R))$ be the local growth rate per step, with $r$ the spectral radius. For $b=0.9$ the eigenvalues of $J_0$ are complex with modulus $\sqrt b$ for $7\times10^{-4}<R<0.9993$, so in this band
\begin{equation}\label{eq:lambdabelow}
\Lambda(R)=\tfrac12\log b=\RampLamBelow,
\end{equation}
independent of $R$. At $R\simeq0.999307$ the eigenvalues coalesce; between this value and the edge they are real and stable, and at $R=1$ one of them passes through $-1$ (period doubling), and above the edge $\Lambda$ rises steeply ($\Lambda=\RampLamA$ at $R=1.01$ and $\RampLamB$ at $R=1.1$). If $R$ changes slowly, $R_t=R_{\rm in}+\rho t$ with $\rho\ll1$, the envelope of the deviation follows the local rate:
\begin{equation}\label{eq:wkb}
|w_t|\;\approx\;C\,|w_{\rm in}|\,\exp\Big(\sum_{s<t}\Lambda(R_s)\Big)\;\approx\;C\,|w_{\rm in}|\,\exp\Big(\frac1\rho\int_{R_{\rm in}}^{R_t}\Lambda(R)\,dR\Big),
\end{equation}
to leading order as $\rho\to0$, for generic initial phases; the prefactor $C$ (turning eigenvectors, phase of the rotation) is not uniformly of order one near the coalescence, but changes $\rho\log|w_t|$ only by $O(\rho|\log\rho|)$. The deviation is back at its initial size when the integral returns to zero, which defines the exit point $R_{\rm out}$ through Neishtadt's \emph{entry--exit relation}
\begin{equation}\label{eq:entryexit}
\int_{R_{\rm in}}^{R_{\rm out}}\Lambda(R)\,dR=0 .
\end{equation}
As $\rho\to0$ the delay $R_{\rm out}-1$ stays finite. To leading order it depends only on where the sweep began, not on $\rho$ or on the size of the deviation: \eqref{eq:entryexit} gives $R_{\rm out}=1.080$ for $R_{\rm in}=0.5$ and $\RampPsiDeep$ for $R_{\rm in}=\RampRinDeep$. The contraction accumulated below the edge must be repaid above it. Neishtadt proved that in analytic systems the delay persists even when the equilibrium moves with the parameter (in merely smooth systems it generally does not); in our map the minimum is a fixed point for every $R$, so the linear delay carries over directly. For slow passage through a Hopf bifurcation it is known as the memory effect \citep{baer1989slow}, and swept period doubling in maps was studied by \citet{kapral1985nonautonomous} and \citet{baesens1991sweep}.

\paragraph{Noise erases the memory.} Now add a forcing $\zeta_t$ of strength $\sigma$, $w_{t+1}=J_0(R_t)w_t+\zeta_t$, which models minibatch noise, rounding or coupling to the rest of the network. A scalar closure that ignores the rotation and non-normality of $J_0$ gives $s_{t+1}^2=e^{2\Lambda(R_t)}s_t^2+\sigma^2$; in the complex band $e^{2\Lambda}=b$, so the variance contracts by $b$ per step and settles at $s_\infty^2=\sigma^2/(1-b)$. This is only a heuristic noise level: the variance of the two-dimensional deviation depends on how the noise enters and grows as $R\to1$ (for gradient noise of variance $\sigma^2$, as in \eqref{eq:ramprec} below, the stationary displacement variance at fixed $R<1$ is $\eta\sigma^2/[2\lambda(1-R)]$). Above the edge the deviation grows from the noise level $s_\times$ reached at the crossing and reaches a threshold $x_c$ at approximately
\begin{equation}\label{eq:noisyexit}
\int_{1}^{R_{\rm out}}\Lambda(R)\,dR\;\approx\;\rho\,\log\frac{x_c}{s_\times}.
\end{equation}
The exit then no longer depends on $R_{\rm in}$: the noise has erased the memory. A deviation that is still above $s_\times$ when the edge is crossed keeps its memory, so the exit follows \eqref{eq:entryexit} for weak noise and \eqref{eq:noisyexit} for strong noise. \Cref{eq:noisyexit} is a heuristic post-edge estimate for $s_\times<x_c$, obtained by amplifying the fluctuation present at the crossing and neglecting later noise injection; the delay it gives grows with the sweep rate $\rho$ and with $\log(1/\sigma)$. At fixed nonzero noise, sufficiently slow sweeps may instead reach the threshold at or before the edge. The quantitative comparisons below use Monte Carlo runs of the linear model \eqref{eq:ramprec}.

\paragraph{A controlled test in a network.} We impose the sweep in a network; the comparison uses no fitted coefficients, and the linear model is driven by the measured curvature path. We take the tanh MLP of \cref{app:networks} (784--128--128--10, MSE loss, 2{,}000 Fashion-MNIST images) in double precision, pretrain it with full-batch EMA momentum at $R=0.5$ for 1{,}500 steps, and then set the step at every iteration to $\eta_t=(R_{\rm in}+\rho t)\,\aeos/\hat\lambda_t$, where $\hat\lambda_t$ is the latest estimate of $\lambda_{\max}$ from warm-started power iteration, refreshed every five steps; the realized $R_t=\eta_t\lambda_{\max}(t)/\aeos$ therefore follows the linear ramp only approximately. A second copy of the network starts at $\theta+\delta u$, with $u$ the top Hessian eigenvector, and takes the same steps; no kinetic map is used, and $x_t=\langle\theta^{(2)}_t-\theta^{(1)}_t,u_t\rangle$ is the deviation along the top mode. The prediction is the linear recurrence for the top mode in the network's own coordinates, the displacement $x_t$ and the corresponding momentum component $d_t$ (the projection of $m_t^{(2)}-m_t^{(1)}$ on the top eigenvector):
\begin{equation}\label{eq:ramprec}
d_{t+1}=b\,d_t+(1-b)\big(\lambda_t x_t+\xi_t\big),\qquad x_{t+1}=x_t-\eta_t\,d_{t+1},
\end{equation}
with the steps $\eta_t$ actually used, $\lambda_t$ interpolated linearly between the eigenvalue estimates, the same initial deviation ($x_0=\delta$, $d_0=0$) and $\xi_t=0$ without injected noise. Changes of the eigenvector between estimates are neglected. This is the exact matrix product of the linear model along the measured path; \eqref{eq:wkb} and \eqref{eq:entryexit} are its slow-sweep approximations. In \RampKickValid{} ramps (three seeds, $R_{\rm in}$ from $0.1$ to $0.97$, $\rho$ from $0.001$ to $0.1$) the network follows the recurrence \eqref{eq:ramprec} while the deviation decays by up to \RampOrdersDown{} orders of magnitude and grows back (\cref{fig:ramp}a): the largest difference, relative to the local amplitude, is \RampKickSeriesErr\% in the median run and at most \RampKickSeriesErrMaxFast\% for $\rho\ge0.003$, and the exit occurs at the same step in network and model in every run. Details, including three excluded runs and the larger errors at $\rho=0.001$, are in \cref{app:delay}.

The exit point clearly remembers the entry point (\cref{fig:ramp}b): at $\rho=0.003$ a sweep entering at $R_{\rm in}=\RampRinDeep$ leaves at $R_{\rm out}=\RampExitDeep$, one entering at $\RampRinShallow$ at $\RampExitShallow$, against $\RampPsiDeep$ and $\RampPsiShallow$ from \eqref{eq:entryexit}. For entries up to $R_{\rm in}=0.9$ the differences from \eqref{eq:entryexit} are at most \RampPsiDiffThree{} at $\rho=0.003$ and \RampPsiDiffOne{} at $\rho=0.001$; they are finite-$\rho$ corrections, which shrink as $\rho\to0$ and which the model on the measured path reproduces exactly. Adding isotropic gradient noise of standard deviation $\sigma$ per coordinate to the second copy (it enters \eqref{eq:ramprec} as $\xi_t$, multiplied by $1-b$ in the momentum; the model distributions are Monte Carlo runs of \eqref{eq:ramprec} with the same $\sigma$) erases the memory as predicted (\cref{fig:ramp}c): with $\sigma=10^{-9}$ the exit still moves with the entry point (from \RampKnLoDeep{} to \RampKnLoShallow), with $\sigma=10^{-7}$ it lies at \RampKnMidDeepLo--\RampKnMidDeepHi{} for every entry up to $R_{\rm in}=0.5$, and with $\sigma=10^{-5}$ the noise alone is as large as the kick, so the exit comes as soon as the edge is crossed. With noise alone (\cref{fig:ramp}d) the exit grows with $\rho$ and with $\log(1/\sigma)$, as \eqref{eq:noisyexit} predicts; \RampKnInBand{} of \RampKnRuns{} network exits with kick and noise, and \RampNoiseInBand{} of \RampNoiseRuns{} with noise only, lie within the 10--90\% range of the model's distribution. Near this well-trained state, the network's top mode behaves almost exactly like the linear two-dimensional map of the theory. The test shows that the delay follows Neishtadt's mechanism quantitatively; it does not show that the top mode of every network is this linear.

\begin{figure}[htbp]
\centering
\includegraphics[width=\linewidth]{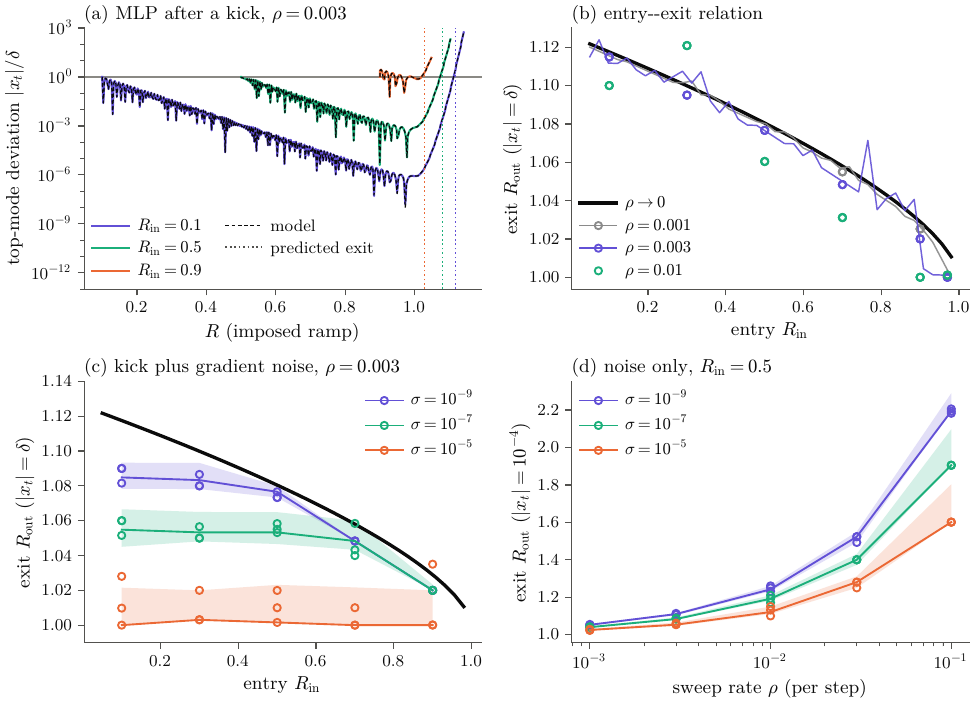}
\caption{Controlled sweeps through the edge in a network (tanh MLP, Fashion-MNIST, full-batch EMA momentum, $\beta_1=0.9$, no map; three seeds). (a) Deviation along the top Hessian direction after a kick $\delta=10^{-6}$ at $R_{\rm in}$, while $R$ is ramped at $\rho=0.003$ per step (solid: network, seed 0; dashed: linear recurrence \eqref{eq:ramprec} on the measured path; dotted: exit predicted by the slow-sweep relation \eqref{eq:entryexit}). (b) Exit point (deviation back at $\delta$) versus entry point: network (markers, $\rho\le0.01$), model on an exact linear ramp (thin lines, $\rho\le0.003$) and the $\rho\to0$ relation \eqref{eq:entryexit} (thick line). (c) Kick plus gradient noise $\sigma$, $\rho=0.003$: network (markers) and model median and 10--90\% range (lines, bands); the thick line is \eqref{eq:entryexit}. (d) Noise only, $R_{\rm in}=0.5$, exit at $|x_t|=10^{-4}$, versus sweep rate.}
\label{fig:ramp}
\end{figure}

\paragraph{Sweeps in ordinary training.} While the oscillation is still small the kinetic map does not act ($g\approx1$), so $\Reff=Rg$ follows $R$ above one until the orbit forms; unless the map is already engaged by the descent momentum, the delay therefore decides how far above the edge the effective step can go. In the reduced model with the cubic map, the peak of $\Reff$ grows with the sweep rate, and the linear theory above predicts it within \DelayPredErrSlow\% for $\rho<0.01$ (\cref{app:delay}). In the networks of \cref{sec:networks}, where the sharpness rises by itself, the overshoot of ResNet-32 grows with the measured sharpening rate (Spearman correlation \DelayResSpearman{} over \DelayResN{} runs, $p=\DelayResSpearmanP$), as predicted, but it is larger than the reduced model gives (\cref{fig:delay} in \cref{app:delay}); we have not identified what sets its size.

\section{Global attraction and an exact quadratic test}\label{sec:global}
\subsection{How much elasticity is enough}
\Cref{thm:main} is local. Whether a trajectory that starts near the unstable minimum actually reaches the orbit is a global question, and with momentum the answer differs between maps that are locally equivalent. To separate the shape of the map from everything else, we use the \emph{Legendre power family}: the speed $s$ solves $s+s^n=z$, which is the velocity of the Lagrangian $L=\tfrac12v^2+\tfrac1{n+1}|v|^{n+1}$ ($n=2$: cubic, $n=3$: quartic). Its elasticity decreases from $1$ to $p=1/n$, and the orbit has the closed form
\begin{equation}\label{eq:powerorbit}
s^{\ast\,n-1}=\tfrac A2-1,\qquad z^\ast=s^\ast\tfrac A2,\qquad \varepsilon^\ast=\frac{A}{2+n(A-2)} ,
\end{equation}
so every member is locally stable for all $A>2$. For each $\beta_1\in\{0,0.5,0.8,0.9,0.95,0.99\}$, $p\in[0.1,0.67]$ and $A\in\{3,\dots,100\}$ (up to $50\times$ the edge step) we iterate \eqref{eq:reduced} from \BasinNZero{} initial conditions ($|y_0|$ log-uniform in $[10^{-4},30]$ in units of $c$, $v_0=0$) and record the share that ends on the symmetric orbit (\cref{tab:basin,fig:basin}). This share is an empirical attraction fraction under this sampling law, not a two-dimensional basin area.

\begin{observation}\label{obs:threshold}
(a) Without momentum, every member of the family attracted all sampled initial conditions at every tested step size. (b) With momentum, the worst-case share over step sizes drops sharply when $p$ falls below a value that grows with $\beta_1$: members with $p\ge1/3$ keep at least $\BasinQuarticEight$ of the starts for $\beta_1\le0.95$, but at $\beta_1=0.99$ the quartic map keeps only \BasinQuarticNineNine{}; members with $p\le1/4$ keep at most \BasinPQuarterNine{} at $\beta_1=0.9$ and \BasinPQuarterNineNine{} at $\beta_1=0.99$. Only $p\ge1/2$ attracted every start in every tested configuration. (c) At $\beta_1=0.9$ the share depends mainly on the orbit elasticity, across the power, relativistic and log-speed families (\cref{fig:basin}a, \cref{fig:quad}d): it is at least \BasinNineHiMin{} whenever $\varepsilon^\ast\ge0.34$ and at most \BasinNineLoMax{} whenever $\varepsilon^\ast<0.25$, with a family-dependent transition in between.
\end{observation}
Starts that miss the orbit end on period-four, higher-period or chaotic attractors near the orbit amplitude, or, for a few per cent of the largest starts, on large-amplitude cycles. Saturating maps ($\varepsilon\to0$: relativistic, clipping, log-speed) lose the orbit at moderate steps: the relativistic map keeps at most \BasinRelNine{} of the starts at $\beta_1=0.9$ once $A\ge4$. Among the members that were robust everywhere ($p\ge1/2$), the cubic-Legendre map damps large momenta the most, and its speed has a closed form. Robustness has a price: on the orbit $g(z^\ast)=2/A$ with $g(z)\sim z^{p-1}$, so for large $A$ and equal $c$ the amplitude grows as $z^\ast\propto A^{1/(1-p)}$, i.e.\ as $A$ for the relativistic map, $A^{3/2}$ for the quartic and $A^{2}$ for the cubic map (near the edge, for $2<A\lesssim3.7$, the cubic amplitude is the smallest of the three). These statements hold for the sampling law above, whose upper end exceeds the orbit amplitude $|y^\ast|=z^\ast/\kappa$ only at small $A$ and $\beta_1$ (for the cubic map at $A=3$, $|y^\ast|=0.75$ in units of $c$ without momentum, $14$ at $\beta_1=0.9$ and $149$ at $\beta_1=0.99$); larger starts are discussed in \cref{app:reduced}.

\begin{table}[htbp]
\centering\small
\caption{Legendre power maps in the reduced model: worst share of \BasinNZero{} initial conditions that reach the symmetric orbit, over seven step sizes $A\in[3,100]$. $p=1/2$ is the cubic-Legendre map, $p=1/3$ the quartic-Legendre map.}
\label{tab:basin}
\begin{tabular}{lcccccccc}
\toprule
& \multicolumn{8}{c}{worst share over $A\in\{3,4,6,10,20,40,100\}$ for $p=1/n=$}\\
$\beta_1$ & 0.67 & 0.50 & 0.40 & 0.33 & 0.25 & 0.20 & 0.14 & 0.10\\
\midrule
0.0 & \textbf{1} & \textbf{1} & \textbf{1} & \textbf{1} & \textbf{1} & \textbf{1} & \textbf{1} & \textbf{1}\\
0.5 & \textbf{1} & \textbf{1} & \textbf{1} & \textbf{1} & 0.98 & 0.96 & 0.96 & 0.51\\
0.8 & \textbf{1} & \textbf{1} & 0.99 & 0.98 & 0.97 & 0.47 & 0.25 & 0.15\\
0.9 & \textbf{1} & \textbf{1} & \textbf{1} & \textbf{1} & 0.48 & 0.15 & 0.11 & 0.11\\
0.95 & \textbf{1} & \textbf{1} & \textbf{1} & 0.99 & 0.11 & 0.05 & 0.03 & 0.04\\
0.99 & \textbf{1} & \textbf{1} & 0.86 & 0.76 & 0.01 & 0.00 & 0.00 & 0.00\\
\bottomrule
\end{tabular}

\end{table}

\begin{figure}[htbp]
\centering
\includegraphics[width=\linewidth]{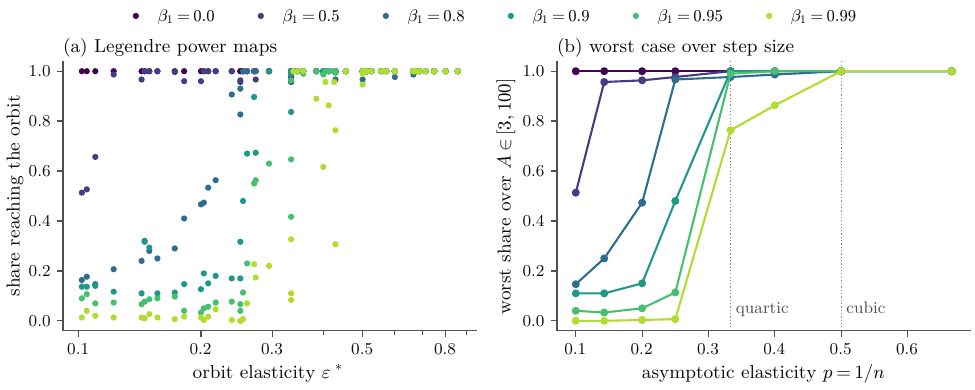}
\caption{Legendre power maps in the reduced model. (a) Share of initial conditions that reach the symmetric orbit versus the orbit elasticity, all $(p,A)$. (b) Worst share over $A\in[3,100]$ versus the asymptotic elasticity $p$.}
\label{fig:basin}
\end{figure}

\subsection{An exact quadratic at network scale}\label{sec:quadratic}
To test the theory in high dimension without self-stabilization, we use linear least squares on 2{,}000 real Fashion-MNIST images \citep{xiao2017fashion}: $L(W)=\tfrac1{2N}\|XW-Y\|_F^2$ with $W\in\R^{785\times10}$ (\QuadParams{} parameters, one tensor). The Hessian is $C\otimes I_{10}$, with $\lambda_{\max}=\QuadLamMax$, and it is constant in time. We use EMA momentum ($\beta_1=0.9$) with $\eta=\mu\cdot38/\lambda_{\max}$, so that $A=2\mu$; the map acts on the tensor norm of the momentum, with $c$ set to $5\%$ of the initial gradient norm of the first seed (used for all seeds); 6{,}000 steps, three seeds, 16 values of $\mu\in[0.5,12]$.

Plain momentum diverges in all \QuadNoneDiv{} runs with $\mu>1$. Above the edge, the quartic and cubic maps reach the symmetric orbit in \QuadQuarticOn{} and \QuadCubicOn{} runs, with $\Reff=1$ and amplitudes $\|m\|/c$ equal to the closed forms of \cref{cor:closed} within relative error \QuadErrQuartic{} and \QuadErrCubic{} (\cref{fig:quad}c). The VR map reaches the orbit in \QuadVROnStable{} runs inside its local stability range $1<\mu<2$ (the others end on coexisting large-amplitude chaotic attractors) and is chaotic beyond it (alternation index \QuadVRAltHigh{} for $\mu\ge2.25$, against $1$ on the orbit). The relativistic map, locally stable everywhere, reached the symmetric orbit in \QuadRelOn{} runs with $\mu\ge2$: \QuadRelPfour{} of them end on a period-four orbit with $\Reff\approx0.5$ and the rest on other periodic orbits, the failure predicted by \cref{obs:threshold}. Below the edge, \QuadSubBad{} quartic and cubic runs (quartic, $\mu=0.9$) ended on a coexisting large-amplitude cycle instead of the minimum, an instance of the large-start caveat.

\begin{figure}[htbp]
\centering
\includegraphics[width=\linewidth]{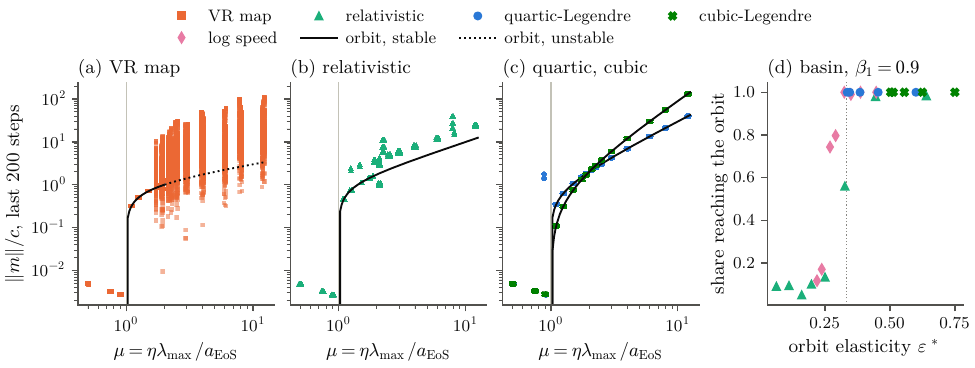}
\caption{Exact quadratic (\QuadParams{} parameters, real data). (a--c) Late momentum norm $\|m\|/c$ versus $\mu=\eta\lambda_{\max}/\aeos$ (markers: last 200 steps of three seeds) against the closed-form orbit (solid: stable, dotted: unstable). (d) Reduced model, $\beta_1=0.9$: share of initial conditions that reach the orbit versus the orbit elasticity, for four families of maps.}
\label{fig:quad}
\end{figure}

\section{Neural-network dynamics}\label{sec:networks}
\subsection{Full-batch CIFAR-10 experiments}\label{sec:fullbatch}
We train with full-batch EMA momentum on a fixed objective: a balanced 1{,}000-image CIFAR-10 subset, a tanh MLP (3072--128--128--10, 10 seeds, $\mu\in\{0.5,1,2,4,8\}$) and a ResNet-32 with batch normalization without running statistics, evaluated on four fixed blocks so that the loss is one deterministic function of the parameters (5 seeds, $\mu\in\{0.5,2,4\}$). We compute $\lambda_{\max}$ at initialization by Lanczos, set $\eta=\mu\cdot38/\lambda_0$, and track the top eigenpair and the per-tensor step multipliers. In a network each tensor $B$ has its own multiplier $g_B$, so we measure the effective ratio along the unit top eigenvector $u$ with the mode-weighted multiplier $\bar g=\sum_B g_B\|u_B\|^2$:
\begin{equation}\label{eq:reffnet}
\Reff=R\,\bar g=R\sum_B g_B\,\|u_B\|^2,\qquad R=\eta\lambda_{\max}/\aeos,
\end{equation}
where $u_B$ is the part of $u$ in tensor $B$. This is the Rayleigh quotient $u^\top\eta GHu/\aeos$ with $G=\mathrm{diag}(g_B)$; if the top eigenvalue is well separated, it equals the top eigenvalue of $\eta GH/\aeos$ to first order in the spread of the multipliers, and it equals $R\,g$ exactly when $u$ lies in one tensor or when all multipliers are equal. $GH$ is not the full linearization of the update with the map: the derivative of \eqref{eq:map} is $g(z)$ across $q$ but $\phi'(z)=g(z)\varepsilon(z)$ along $q$, which is why stability involves the elasticity. $\Reff\approx1$ therefore tests the balance of the effective step on the orbit (\cref{cor:gated}); it does not by itself measure a Floquet multiplier of the network. We compare plain momentum with the quartic, cubic, relativistic and VR maps at two scales ($c=0.05$ absolute; $c=2\|\nabla_BL(\theta_0)\|$ relative). The questions, fixed before the runs, were whether $\Reff$ is close to one for the convex maps while the sharpness is above the edge, whether the saturating and VR maps fail as predicted, whether the end state returns to $R\approx1$, and how the scale $c$ matters. We summarize completed runs by their diagnostics with $R\ge1.1$ after step 100; these two thresholds were selected after a first inspection of the trajectories, and the summary hardly depends on them (\cref{app:networks}). Results on a smaller Fashion-MNIST MLP and the full per-map results are also in \cref{app:networks}.

\paragraph{MLP} (\EMlpRuns{} runs; \cref{fig:e1mlp}). Plain momentum diverges in \EMlpMomDivFour{} runs with $\mu\ge4$ and in \EMlpMomDivTwo{} at $\mu=2$; with the absolute scale, no run with a kinetic map diverged. The convex maps obey the effective-edge law above the plain edge: their points in \cref{fig:e1mlp}b lie on the line $\bar g=1/R$, and over the \EMlpCvxN{} runs that spend time above the edge the per-run median of $\Reff$ is \EMlpCvxMed{} (IQR \EMlpCvxLo{}--\EMlpCvxHi{}). In the run of \cref{fig:e1mlp}a, with the quartic map, the sharpness reaches $1.75\times$ the plain edge, and from then on $\Reff$ stays between $0.99$ and $1.02$. The VR map behaves as its non-convex profile predicts: it is caught on the branch where a larger momentum gives a smaller step, so its points lie below the line, with $\Reff=\EMlpVRMed$ (IQR \EMlpVRLo--\EMlpVRHi), and the sharpness of the network rises far above the edge (\cref{fig:e1mlp}a). In the end the network's own self-stabilization returns the sharpness to the plain edge, and the multiplier to one, for the quartic and relativistic maps; the map therefore acts during the transient. With the relative scale the maps engage only once the oscillation is large: the open triangles near $\bar g=1$ in \cref{fig:e1mlp}b are diagnostics at which the map has not yet acted although the sharpness is above the edge. From $\mu=4$ on these runs fail, while every map with the absolute scale trains at every $\mu$ (\cref{fig:e1mlp}c).

\begin{figure}[htbp]
\centering
\includegraphics[width=\linewidth]{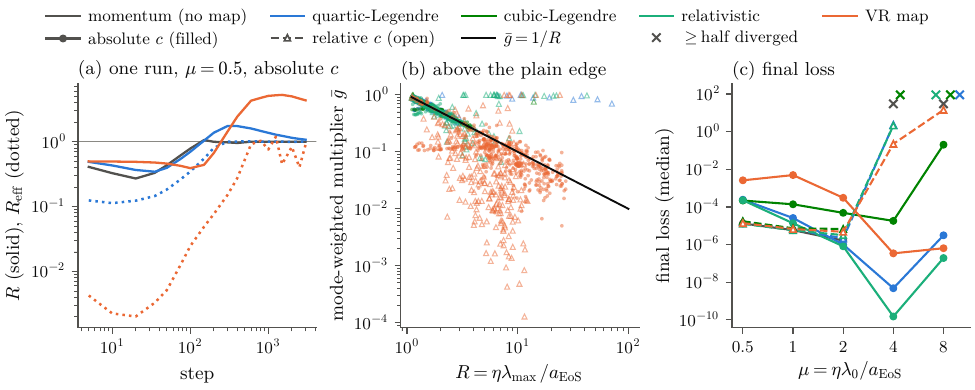}
\caption{Full-batch EMA momentum ($\beta_1=0.9$) on a tanh MLP, 1{,}000 CIFAR-10 images, 10 seeds per setting. Colour: map; filled circles and solid lines: absolute scale $c=0.05$; open triangles and dashed lines: relative scale $c=2\|\nabla_BL(\theta_0)\|$. (a) One seed at $\mu=0.5$, absolute scale: sharpness ratio $R=\eta\lambda_{\max}/\aeos$ (solid) and effective ratio $\Reff$ of \eqref{eq:reffnet} (dotted). Plain momentum stays at the edge, $R\approx1$. With the quartic map the sharpness rises to $1.75$ times the edge while $\Reff$ stays at one; with the VR map the sharpness rises to more than five times the edge while $\Reff$ mostly stays below one. (b) Mode-weighted multiplier $\bar g$ of \eqref{eq:reffnet} against $R$ at every diagnostic with $R\ge1.1$ after step 100, all completed runs. Points on the line $\bar g=1/R$ satisfy $\Reff=1$: the convex maps lie on it, the VR map lies below it (a smaller step than the edge requires), and the open triangles near $\bar g=1$ are relative-scale diagnostics at which the map has not yet engaged, so that the effective step is above the edge. (c) Median final loss versus step size; crosses mark settings in which at least half of the runs diverged.}
\label{fig:e1mlp}
\end{figure}

\paragraph{ResNet-32} (\EResRuns{} runs; \cref{tab:e1res,fig:e1res}). Here the outcome is decided by whether the map engages before the oscillation grows. At $\mu\ge2$, plain momentum either diverged (\EResMomHighDiv{} of \EResMomHighN{} runs) or catapulted towards a nearly constant predictor, ending at losses \EResMomCollLo--\EResMomCollHi{} (predicting the class mean gives $0.09$). All \EResAbsHighN{} runs with a map at the absolute scale, the VR map included, trained, and only \EResRelHighTrained{} of the \EResRelHighN{} runs at the relative scale did. Unlike in the MLP, the effective-edge law does not describe the ResNet transient: the sharpness rises above the edge and falls back within one or two hundred steps, too fast for the oscillation to form (\cref{sec:delay}). In this network the map works as an early brake on the descent momentum, not as a regulator of the edge oscillation.

\begin{table}[htbp]
\centering\small
\caption{ResNet-32 (5 seeds): runs that trained (final loss $<10^{-3}$), collapsed (final loss $>0.03$, near the constant-predictor value $0.09$) or diverged; (+1): one run in between.}
\label{tab:e1res}
\begin{tabular}{lccc}
\toprule
& \multicolumn{3}{c}{trained / collapsed / diverged}\\
method & $\mu=0.5$ & $\mu=2$ & $\mu=4$\\
\midrule
momentum (no map) & 5/0/0 & 0/4/1 & 0/1/4\\
quartic-Legendre, abs. $c$ & 5/0/0 & 5/0/0 & 5/0/0\\
cubic-Legendre, abs. $c$ & 5/0/0 & 5/0/0 & 5/0/0\\
relativistic, abs. $c$ & 5/0/0 & 5/0/0 & 5/0/0\\
VR map, abs. $c$ & 5/0/0 & 5/0/0 & 5/0/0\\
quartic-Legendre, rel. $c$ & 5/0/0 & 0/5/0 & 0/3/2\\
cubic-Legendre, rel. $c$ & 5/0/0 & 2/3/0 & 0/2/3\\
relativistic, rel. $c$ & 5/0/0 & 0/5/0 & 0/4/1\\
VR map, rel. $c$ & 5/0/0 & 0/3/1 (+1) & 0/3/2\\
\bottomrule
\end{tabular}

\end{table}

\begin{figure}[htbp]
\centering
\includegraphics[width=\linewidth]{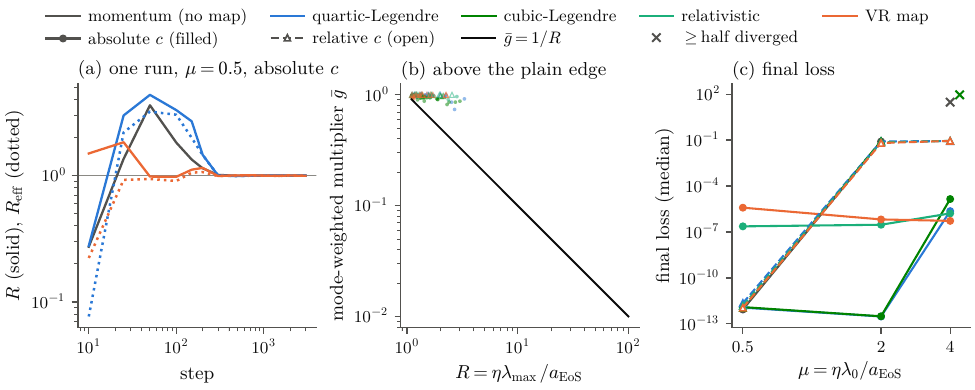}
\caption{Full-batch EMA momentum on ResNet-32 (batch normalization without running statistics), 1{,}000 CIFAR-10 images, 5 seeds; panels as in \cref{fig:e1mlp}. In (b) the multipliers stay near one above the plain edge: the map has not yet engaged.}
\label{fig:e1res}
\end{figure}

\subsection{Adam: maps on the preconditioned direction}\label{sec:adam}
For Adam, a kinetic map can act on the preconditioned direction $q=\hat m/(\sqrt{\hat s}+\epsilon)$ or on the raw first moment $\hat m$. The two differ at the edge:
\begin{proposition}[Preconditioned maps are amplitude-blind]\label{prop:blind}
Let a coordinate of the gradient alternate as $\nabla_t=(-1)^tY$ for longer than the second-moment memory $1/(1-\beta_2)$. Then Adam's $\hat m_t\to\pm\kappa Y$ and $\sqrt{\hat s_t}\to Y$, so $q_t\to\pm\kappa\,Y/(Y+\epsilon)$ with $\kappa=(1-\beta_1)/(1+\beta_1)$: independent of $Y$ for $\epsilon=0$, and approximately so for $Y\gg\epsilon$.
\end{proposition}
In this stationary limit a map on the preconditioned direction applies an attenuation that does not depend on the amplitude: it can change the oscillation, but it cannot provide feedback on its absolute size, and it reacts mainly to amplitude changes faster than $1/(1-\beta_2)$ steps. A map on the raw first moment (in absolute units, as in VRAdam) retains that information. We compare both in full-batch Adam on the ResNet-32 objective of \cref{sec:networks}. With relativistic maps on $\mathrm{rms}(q_B)$ and on $\|\hat m_B\|$ (\cref{tab:e2}; three seeds, $\eta\in\{0.002,0.01\}$), all \ETwoRuns{} runs end at the adaptive edge (late preconditioned $R$ between \ETwoRlateLo{} and \ETwoRlateHi), and none diverges. The maps act early, while the descent momentum is large (smallest multiplier in the first 100 steps: \ETwoMgateEarlyMin--\ETwoMgateEarlyMax{} for the map on $\hat m$, \ETwoVrEarlyMin--\ETwoVrEarlyMax{} for VRAdam; the map on $q$ starts at \ETwoQgateEarlyMin{} because $q=\pm1$ at the first step), and they are essentially inactive at the edge: in the second half of training the smallest multiplier is \ETwoMgateLateMin{} for the map on $\hat m$ and \ETwoVrLateMin{} for VRAdam. Adam's edge oscillation is therefore small in absolute units, whereas the map on $q$ still reacts to transients (down to \ETwoQgateLateMin), as \cref{prop:blind} implies.

\begin{table}[htbp]
\centering\small
\caption{Full-batch Adam on the ResNet-32 objective, medians over three seeds. $R$ is the preconditioned sharpness ratio $\eta\lambda_{\max}(P^{-1}H)/38$; ``late'' is the median over steps $\ge1500$.}
\label{tab:e2}
\resizebox{\linewidth}{!}{\begin{tabular}{llccccc}
\toprule
$\eta$ & method & peak $R$ & late $R$ & min.\ mult.\ (steps $\ge1500$) & min.\ mult.\ (steps $<100$) & final loss\\
\midrule
0.002 & Adam & 1.29 & 0.99 & 1.000 & 1.00 & $3.4{\times}10^{-6}$\\
0.002 & VRAdam & 1.08 & 0.96 & 1.000 & 0.59 & $2.4{\times}10^{-5}$\\
0.002 & Adam + relativistic map on $\mathrm{rms}(q)$, $c=0.5$ & 1.13 & 1.00 & 0.994 & 0.45 & $5.4{\times}10^{-6}$\\
0.002 & Adam + relativistic map on $\|\hat m\|$, $c=0.05$ & 1.36 & 0.99 & 0.996 & 0.15 & $4.9{\times}10^{-6}$\\
0.01 & Adam & 1.08 & 0.96 & 1.000 & 1.00 & $1.7{\times}10^{-4}$\\
0.01 & VRAdam & 1.26 & 0.94 & 1.000 & 0.30 & $3.1{\times}10^{-6}$\\
0.01 & Adam + relativistic map on $\mathrm{rms}(q)$, $c=0.5$ & 1.19 & 0.96 & 0.918 & 0.45 & $1.6{\times}10^{-4}$\\
0.01 & Adam + relativistic map on $\|\hat m\|$, $c=0.05$ & 1.30 & 0.93 & 0.998 & 0.15 & $3.1{\times}10^{-5}$\\
\bottomrule
\end{tabular}
}
\end{table}

In minibatch training with Adam (ResNet-32, batch 1{,}000), the preconditioned sharpness of the training batch stays well below the full-batch threshold, while a preconditioned batch sharpness scales as $1/\eta$, consistent with dynamics limited by stochastic stability \citep{andreyev2024eoss}; see \cref{app:networks}.

\section{Optimizer benchmarks and effective-rate controls}\label{sec:practice}
\begin{algorithm}[htbp]
\caption{Kinetic Adam (per-tensor map)}\label{alg:kadam}
\begin{algorithmic}[1]
\Require lr $\eta_t$, $(\beta_1,\beta_2)$, $\epsilon$, weight decay $\omega$, step multiplier $g$, scale $c$, source $\in\{q,\hat m\}$
\For{each step and each parameter tensor $\theta_B$ with gradient $\nabla_B$}
\State $m_B\gets\beta_1m_B+(1-\beta_1)\nabla_B$; $s_B\gets\beta_2s_B+(1-\beta_2)\nabla_B^2$; $q_B\gets\hat m_B/(\sqrt{\hat s_B}+\epsilon)$
\State $r_B\gets\mathrm{rms}(q_B)/c$ (source $q$) or $\mathrm{rms}(\hat m_B)/c$ (source $\hat m$)
\State $\theta_B\gets\theta_B(1-\eta_t\omega)-\eta_t\,q_B\,g(r_B)$
\EndFor
\end{algorithmic}
\end{algorithm}
\Cref{alg:kadam} adds one reduction per tensor to AdamW; with $g\equiv1$ it is AdamW exactly. The map can wrap any optimizer in update space; we also use it on Muon. \Cref{prop:blind} favours a map on the raw first moment; in the CNN, character-level and ResNet-18 benchmarks below, the Kinetic Adam variants act on $q$ (with a fixed or an EMA scale), and in the GPT and grokking studies on $\hat m$.

\paragraph{Benchmarks.} With an equal tuning budget (10 random-search trials per method, three fresh seeds; \cref{app:practice}), every Adam variant with a fixed scale, VRAdam and the preconditioned curvature controller end close to AdamW (test loss $0.266$--$0.280$ against $0.275\pm0.004$ on a Fashion-MNIST CNN, and $1.70$--$1.74$ on a character-level transformer), and Muon with a quartic map ends close to Muon ($0.259$ against $0.261$ on the CNN). VRAdam's one distinct advantage is at very large learning rates: on an MLP it is the only method that still learns at about $30\times$ AdamW's best learning rate (\cref{fig:lrsens} in \cref{app:practice}). This is consistent with braking on the absolute momentum magnitude (\cref{prop:blind}), not from velocity dynamics.

\paragraph{ResNet-18 on CIFAR-10 with effective-learning-rate controls.} After a random search (\cref{app:practice}), Kinetic Muon had the highest accuracy, and the quartic Adam variants raised accuracy but also cross-entropy relative to AdamW, consistent with a larger effective step. We therefore reran the selected configurations next to controls with matched effective learning rates (\cref{tab:e3}; \EThSeeds{} seeds shared by all rows, so differences are paired). These are the seeds of the tuned comparison, so the tuned rows repeat those runs up to non-deterministic GPU arithmetic and are not an independent replication. Kinetic Muon beats tuned Muon, by $\EThKinVsMuonAcc\pm\EThKinVsMuonAccSd$ pp accuracy and $\EThKinVsMuonCE\pm\EThKinVsMuonCESd$ cross-entropy ($p=\EThKinVsMuonAccP$ and $\EThKinVsMuonCEP$, paired $t$-test). Against Muon run at Kinetic Muon's effective learning rates, the differences shrink to $\EThKinVsEffAcc\pm\EThKinVsEffAccSd$ pp and $\EThKinVsEffCE\pm\EThKinVsEffCESd$ (95\% intervals \EThKinVsEffAccLo{} to \EThKinVsEffAccHi{} pp and \EThKinVsEffCELo{} to \EThKinVsEffCEHi; $p=\EThKinVsEffAccP$, $\EThKinVsEffCEP$); in point estimate, the change of learning rate alone gives about half of the cross-entropy difference ($\EThEffVsMuonCE$, $p=\EThEffVsMuonCEP$). For Adam, the global quartic map is not resolved from AdamW at three times its tuned learning rate ($\EThKqVsThreeAcc$ pp, $\EThKqVsThreeCE$ cross-entropy, $p>\EThKqVsThreeMinP$). Matching the effective learning rates thus removes most of the difference for the Adam map and, in point estimate, about half of it for Kinetic Muon; the remaining differences are not resolved by five seeds, so these experiments support a substantial learning-rate contribution without establishing equivalence. All $p$-values are unadjusted paired $t$-tests for the comparisons the controls were designed for.

\begin{table}[htbp]
\centering\small
\caption{ResNet-18 / CIFAR-10, effective-learning-rate controls: mean $\pm$ s.d.\ over \EThSeeds{} seeds, and paired differences to the reference row of each block ($^\ast$: $p<0.05$, paired $t$-test).}
\label{tab:e3}
\resizebox{\linewidth}{!}{\begin{tabular}{lcccc}
\toprule
configuration & test accuracy (\%) & test CE & $\Delta$acc (pp) & $\Delta$CE \\
\midrule
Muon, tuned & 94.02 $\pm$ 0.18 & 0.351 $\pm$ 0.006 & ref. & ref.\\
Kinetic Muon (quartic EMA map), tuned & 94.19 $\pm$ 0.17 & 0.335 $\pm$ 0.008 & $+0.17$ $\pm$ 0.13$^\ast$ & $-0.016$ $\pm$ 0.011$^\ast$\\
Muon at Kinetic Muon's effective lr ($\times0.66$) & 94.09 $\pm$ 0.07 & 0.344 $\pm$ 0.008 & $+0.07$ $\pm$ 0.20 & $-0.008$ $\pm$ 0.007\\
Muon at Kinetic Muon's nominal lr & 94.01 $\pm$ 0.13 & 0.362 $\pm$ 0.009 & $0.00$ $\pm$ 0.22 & $+0.010$ $\pm$ 0.013\\
\midrule
AdamW, tuned & 92.36 $\pm$ 0.19 & 0.338 $\pm$ 0.010 & ref. & ref.\\
AdamW at $3\times$ tuned lr & 92.85 $\pm$ 0.18 & 0.360 $\pm$ 0.007 & $+0.50$ $\pm$ 0.23$^\ast$ & $+0.022$ $\pm$ 0.015$^\ast$\\
K-Adam quartic, global, tuned & 92.91 $\pm$ 0.21 & 0.368 $\pm$ 0.007 & $+0.55$ $\pm$ 0.28$^\ast$ & $+0.030$ $\pm$ 0.014$^\ast$\\
\bottomrule
\end{tabular}
}
\end{table}

\paragraph{Small GPT.} In a \EFiveParams M-parameter GPT-2-style model trained on FineWeb-Edu \citep{penedo2024fineweb} (\cref{app:practice}), every method is best at the same learning rate, where the final validation losses of AdamW, clipping, VRAdam and the cubic map span only \EFiveBestSpread{} nats, against a spread of \EFiveAdamwSeedSpread{} over initialization seeds for AdamW. At about three times that learning rate AdamW fails, clipping largely avoids the failure, the cubic map partly does, and VRAdam does not. With two seeds per setting, no kinetic map improves the best loss beyond seed noise or extends the usable range of learning rates beyond what clipping does.

\section{Discussion and conclusions}\label{sec:discussion}
\subsection{Connections to nonlinear dynamics}\label{sec:connections}
The edge of stability is a phenomenon of nonlinear dynamics, and \cref{sec:delay} applied one classical result to it. Here we note further connections to nonlinear dynamics and report a test on grokking.

\paragraph{Weak dissipation and capture into resonance.} The one-step Jacobian of \eqref{eq:reduced} has determinant $b=\beta_1$ everywhere. In the variables $(v,u)$ with $u=(1-b)y$ the map reads $v'=bv+u$, $u'=u-(1+b)A\,\phi(v')$; as $\beta_1\to1$ at fixed $A$ it becomes a weakly dissipative perturbation of the area-preserving map $v'=v+u$, $u'=u-2A\,\phi(v')$, and the symmetric orbit becomes a weakly damped elliptic point (\cref{app:proofs}). In area-preserving maps, elliptic points are surrounded by chains of resonance islands, and weak dissipation can turn islands into competing attractors. Which attractor a trajectory reaches is then expected to be a problem of capture into resonance, as in slowly varying or weakly dissipative Hamiltonian systems, where capture probabilities are set by phase-space fluxes \citep{neishtadt1975separatrix,arnold2006mathematical} and, for maps with a constant Jacobian, by the areas of the islands \citep{chirikov1981,lieberman1985}. Capture into resonance governs, for example, the acceleration of charged particles by waves \citep{itin2000captures,itin2002trapping,neishtadt2011surfatron,artemyev2018trapping}, the destruction of adiabatic invariance in slowly perturbed elliptic billiards \citep{itin2003elliptic}, Fermi acceleration in time-dependent rectangular billiards \citep{itin2012fermi} and resonant phenomena in the classical three-body Coulomb problem \citep{itin2003coulomb}. We expect the period-four and higher-period attractors of \cref{sec:global} to be of this kind, and the growth with $\beta_1$ of the elasticity required for global attraction to reflect the weaker dissipation; a quantitative treatment is left for future work.

\paragraph{Averaged force.} Averaging the loss over a period-two oscillation of amplitude $\delta$ along the top eigenvector $u$ gives $\tfrac12[L(\theta+\delta u)+L(\theta-\delta u)]=L(\theta)+\tfrac12\delta^2\lambda(\theta)+O(\delta^4)$: the fast oscillation exerts a force $-\tfrac12\delta^2\nabla\lambda$ on the slow variables and pushes them towards lower sharpness. This is the averaged form of self-stabilization \citep{damian2023selfstab}, analogous to the ponderomotive force and to Kapitza's pendulum, where a fast vibration creates an effective potential for the slow motion \citep{kapitza1951,landau1976mechanics}. A kinetic map sets $\delta$, and hence the strength of this force (\cref{app:reduced}).

\paragraph{Other analogies.} The EoS has the structure of self-organized criticality \citep{bak1987soc}, with progressive sharpening as the slow drive and catapults and self-stabilization as the fast relaxation; \citet{katsnelson2021soc} find that learning dynamics are generically attracted to such a critical state, and in this language a kinetic map moves the critical point from $\eta\lambda=\aeos$ to $\eta\lambda g=\aeos$. A kinetic map is also a governor on the optimizer's speed, in the sense of Maxwell's analysis of feedback control \citep{maxwell1868governors}: a convex map throttles less than proportionally, whereas VRAdam's non-convex map throttles more than proportionally beyond $z=1$, so that a larger drive gives a smaller speed. Finally, for networks with normalization layers \citet{kodryan2022three} find convergence, chaotic-equilibrium and divergence regimes as the effective learning rate grows, and \citet{lobacheva2021periodic} describe periodic destabilizations under weight decay; the reduced model has analogous regimes, and with the convex maps we tested ($g_\infty=0$) the divergence regime of plain momentum is replaced by bounded oscillations.

\paragraph{Grokking.} Networks trained on small algorithmic tasks can generalize long after fitting the training set \citep{power2022grokking,nanda2023progress}. With adaptive optimizers this late generalization coincides with norm growth and loss spikes (``slingshots'') \citep{thilak2022slingshot}, and \citet{prieto2025grokking} relate grokking to numerical instability, so edge-of-stability dynamics could be involved. Clipping and kinetic maps on $\hat m$ can damp such spikes, which gives a direct test: if the spikes drove grokking, damping them should delay it. In modular addition with weight decay (\cref{app:grok}), clipping and the cubic map with $k=1$ lowered the median of the largest training loss after memorization about fourfold (not on every seed) without a significant delay of grokking, and larger spikes did not go with earlier grokking. \emph{Conclusion:} in this setting, reducing the measured loss peaks did not consistently delay grokking, and the large spikes seen without weight decay did not produce it; with three seeds per method we see no significant effect of the maps on when grokking occurs. Beyond this negative result we found no link between grokking and the mechanisms studied here.

\subsection{Limitations}
The reduced model freezes Adam's preconditioner; minibatch noise, the scale invariance of batch-normalized networks and the large-amplitude cycles reachable from large starts are outside the theorem, and \cref{obs:threshold} is numerical, not proved. In nonlinear networks the sharpness self-stabilizes to the plain edge (except with the cubic map in our MLP, which ends above it, at $R$ up to \EMlpCubicAbsEndHi{} for $\mu\ge1$), so the orbit theory governs the transient, not the end state. The basin results are for sampled initial conditions on a grid of step sizes and momentum values, and the robustness of the cubic map is a statement about the reduced model: in the MLP it trained worse than the quartic and relativistic maps. The controlled delay test was run near a well-trained state, where the top mode is linear to high accuracy; in ordinary ResNet-32 training the overshoot follows the predicted trend but is larger than the reduced model gives. In ResNet-32 the transient is too fast for the orbit to form, so maps at the absolute scale help there as early brakes, not through the effective-edge law. As optimizers, kinetic maps were within noise of tuned baselines on our small benchmarks, and on ResNet-18 controls with matched effective learning rates reproduce most of the effect of an Adam map and about half of the gain of Kinetic Muon, with the remainder not resolved by five seeds.

\subsection{Recommendations}
(i) Use a convex map, and choose its elasticity knowing the trade-off: a high asymptotic elasticity (cubic, $p=1/2$) gave the most reliable attraction to the orbit in the reduced model, a low one a smaller oscillation at large steps; in our MLP at $\mu\ge1$ the quartic and relativistic maps trained better than the cubic one, while in ResNet-32 the differences were small. (ii) Apply the map per tensor and to the raw first moment. (iii) Choose $c$ small enough that the map engages before the oscillation grows: it sets the oscillation amplitude, and in our MLP study a scale relative to the initial gradient norm engaged too late. (iv) Report step multipliers and effective learning rates with any comparison.

\subsection{Conclusions}
We have analysed how kinetic maps change momentum dynamics after a minimum loses stability. In the reduced model, the elasticity of the speed profile gives an exact criterion for the local stability of the symmetric period-two orbit, and on that orbit the step multiplier compensates exactly for the excess curvature. VRAdam's damper violates this criterion from twice the edge step on, while the Legendre-dual map of its own Lagrangian satisfies it at every step. Numerical basin experiments show that the local result does not determine which attractor a trajectory reaches: with heavy momentum this also requires elasticity at large momentum.

Our central dynamical result concerns the gradual passage through the edge. In controlled network experiments, the linear recurrence on the measured curvature path reproduces the delayed loss of stability without fitted parameters, and the exit points follow the entry--exit relation of Neishtadt's theory of slowly swept bifurcations, including the dependence on the entry point and its suppression by noise. In ordinary full-batch training the overshoot grows with the sharpening rate, as predicted, but it is larger than the reduced model gives, and we have not identified what sets its size.

These dynamical results do not imply a general improvement in predictive performance. In networks the map acts in two different ways: it holds a settled oscillation at the effective edge (MLP), or it brakes the early descent (ResNet-32). On ResNet-18, controls with matched effective learning rates reproduce much of the apparent gains: most of the difference for an Adam map and about half of it for Kinetic Muon. Separating these mechanisms, and stating where each explanation is supported, is a main aim of this work. The weakly dissipative limit $\beta_1\to1$, in which capture into resonance should decide which attractor training reaches, is left for future work.

\subsubsection*{Acknowledgements}
This work used the Dutch national e-infrastructure with the support of the SURF Cooperative using grant no.~EINF-19390. Access to the supercomputer Snellius was provided under the NWO call ``Computation time national computer systems 2026''.

\bibliographystyle{tmlr}
\bibliography{refs}

\appendix
\section{Proofs for Section~\ref{sec:theory}}\label{app:proofs}
\paragraph{Setting.} The map $T(y,v)=(y-a\phi(bv+(1-b)y),\,bv+(1-b)y)$ is odd.
\paragraph{\Cref{thm:main}(i).} At the minimum, $J_0=\begin{psmallmatrix}1-a(1-b)&-ab\\1-b&b\end{psmallmatrix}$; the Jury conditions reduce to $A<2$. Its trace is $1+b-a(1-b)$ and its determinant $b$, which gives the growth rate $\Lambda(R)$ of \cref{sec:delay}.
\paragraph{\Cref{thm:main}(ii).} $T(p)=-p$ gives $v=-\kappa y$, $z=\kappa y$ and $g(z)=2/A$, with a unique root for $g_\infty<2/A<1$ (uniqueness among symmetric orbits only; asymmetric two-orbits exist for the VR map when $A>4$, \cref{prop:vrasym}).
\paragraph{\Cref{thm:main}(iii).} $DT$ at $\pm p$ is the same $J$ because $\phi'$ is even, so $DT^2=J^2$; $\operatorname{tr}J=(1+b)(1-2\varepsilon^\ast)$, $\det J=b$; for $0\le\det<1$, the spectral radius of $J$ is below one if and only if $|\operatorname{tr}J|<1+\det J$. At $\varepsilon^\ast=0$ the eigenvalues are $1$ and $b$.
\paragraph{\Cref{cor:gated}.} Immediate from $g(z^\ast)=2/A$ and $A=\kappa\eta\lambda=2\eta\lambda/\aeos$.
\paragraph{\Cref{cor:closed}.} VR: $g=1/(1+z^2)$. Relativistic: $g=(1+z^2)^{-1/2}$. Quartic: $z=s+s^3$, $g=1/(1+s^2)$, $\varepsilon=(1+s^2)/(1+3s^2)$. Cubic: $z=s+s^2$, $g=1/(1+s)$, $\phi'=1/(1+2s)$, $\varepsilon=(1+s)/(1+2s)=A/(2A-2)$.
\paragraph{\Cref{prop:blocks}.} On the orbit the multiplier is the constant $2/(\eta\lambda_{\max})$, and $\partial\|u\|/\partial u_i=0$ at $u_i=0$.
\paragraph{\Cref{prop:blind}.} For an alternating input the EMA $m_t$ converges to $\pm(1-\beta_1)Y/(1+\beta_1)$ and $s_t$ to $Y^2$; the bias corrections tend to one.
\paragraph{Orbit multipliers (\cref{sec:connections}).} From (iii), the characteristic polynomial $\nu^2-(1+b)(1-2\varepsilon^\ast)\nu+b$ of the one-step Jacobian has complex roots of modulus $\sqrt b$ whenever $(1+b)^2(1-2\varepsilon^\ast)^2<4b$ (for $b=0.9$: $0.0007<\varepsilon^\ast<0.9993$), with $\cos\theta=(1+b)(1-2\varepsilon^\ast)/(2\sqrt b)$; $\theta=\pi/2$ at $\varepsilon^\ast=1/2$. The two-step (Floquet) multipliers are $b\,e^{\pm2i\theta}$. With $u=(1-b)y$: $v'=bv+(1-b)y=bv+u$ and $u'=(1-b)y'=u-(1-b)a\phi(v')=u-(1+b)A\phi(v')$.

\begin{proposition}[Asymmetric two-cycles of the VR map]\label{prop:vrasym}
Let $\phi(z)=z/(1+z^2)$ and $b\in[0,1)$. Besides the symmetric orbit, \eqref{eq:reduced} has two-cycles whose momenta after the two steps are $z_1>0$ and $-z_2<0$ with $z_1\neq z_2$ if and only if $A>4$. On them $z_1z_2=1$ and $(z_1+z_1^{-1})^2=A$, so $z_1=\tfrac12(\sqrt A\pm\sqrt{A-4})$; the two signs give a pair of mirror images under $(y,v)\mapsto(-y,-v)$, which branches off the symmetric orbit ($z_1=z_2=1$) at $A=4$. The two-step Jacobian on these cycles has determinant $b^2$ and trace $1+b^2-2(1+b)^2(A-4)$, so they are strictly linearly stable if and only if $4<A<4+(1+b^2)/(1+b)^2$. Without momentum this is $4<A<5$; the nontrivial multiplier $9-2A$ passes through $-1$ at $A=5$, where the asymmetric cycles period-double.
\end{proposition}
\begin{proof}
For a two-cycle $(y_1,v_1)\mapsto(y_2,v_2)\mapsto(y_1,v_1)$ write $v_2=z_1$, $v_1=-z_2$. Adding the two $y$-updates gives $\phi(z_1)=\phi(z_2)$, and for the VR map this is equivalent to $(z_1-z_2)(1-z_1z_2)=0$, so $z_1z_2=1$ when $z_1\neq z_2$. Subtracting the two $v$-updates gives $(1+b)(v_2-v_1)=(1-b)(y_1-y_2)=(1-b)a\phi(z_1)$, i.e.\ $z_1+z_2=A\phi(z_1)$, which with $z_2=1/z_1$ becomes $(z_1+z_1^{-1})^2=A$; it has a solution $z_1\neq1$ if and only if $A>4$. With $\phi'(z)=(1-z^2)/(1+z^2)^2$ and $z_1+z_1^{-1}=\sqrt A$, $d_1=\phi'(z_1)$ and $d_2=\phi'(z_2)$ satisfy $d_1+d_2=(A-4)/A$ and $d_1d_2=-(A-4)/A^2$. The product of the one-step Jacobians $J(d)=\begin{psmallmatrix}1-a(1-b)d&-abd\\1-b&b\end{psmallmatrix}$ has determinant $b^2$ and trace $1+b^2-(1+b)^2\big[A(d_1+d_2)-A^2d_1d_2\big]=1+b^2-2(1+b)^2(A-4)$, and $|\operatorname{tr}|<1+\det$ gives the stated interval.
\end{proof}

\section{Additional results for the reduced model}\label{app:reduced}
\begin{figure}[t]
\centering
\includegraphics[width=\linewidth]{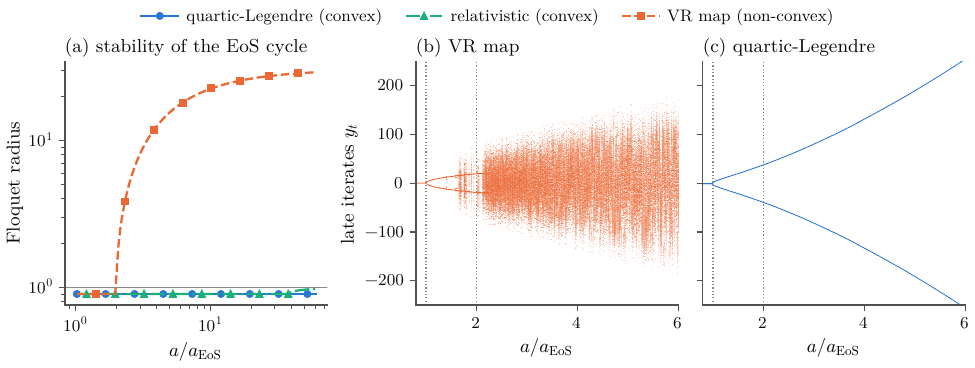}
\caption{EMA momentum, $\beta_1=0.9$. (a) Floquet radius of the symmetric EoS orbit versus $a/\aeos$. (b,c) Late iterates.}
\label{fig:momentum}
\end{figure}

\paragraph{Large starts.} The basin statements of \cref{sec:global} hold for the sampling law stated there, whose starts are mostly smaller than the orbit (at most $29\%$ of them exceed $|y^\ast|$, at $\beta_1=0$ and $A=3$), as for trajectories that leave the minimum as it loses stability. Larger starts can lock into large-amplitude cycles for every map we tested: for the cubic map the smallest such start on a logarithmic grid lies between \LargeICCubicMin{} and \LargeICCubicMax{} times the orbit amplitude $|y^\ast|$ ($\beta_1\in\{0.9,0.99\}$, $A\in\{3,6,20\}$).

\paragraph{Global versus per-block maps.} \Cref{fig:blocks} illustrates \cref{prop:blocks} on a block-diagonal quadratic.

\begin{figure}[t]
\centering
\includegraphics[width=0.66\linewidth]{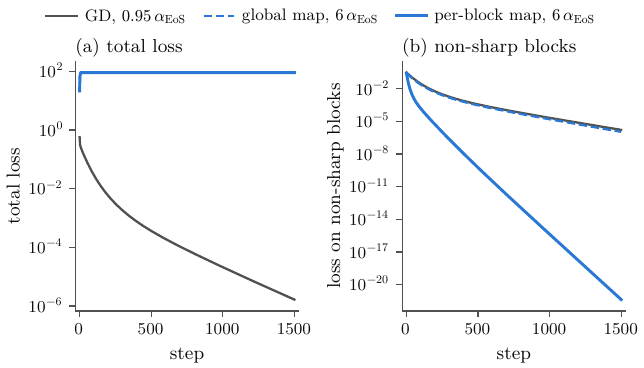}
\caption{Block-diagonal quadratic, quartic map at $6\times$ the edge step: global versus per-block maps.}
\label{fig:blocks}
\end{figure}

\paragraph{Amplitude-controlled self-stabilization.} The map bounds the oscillation instead of removing it. In the model of \citet{damian2023selfstab}, $f=\tfrac12(S_0+\gamma w)x^2$, the momentum amplitude of the oscillation is $c\,z^\ast(A)$ and the sharpness decreases at a rate $\propto c^2$ (measured slopes $\SelfStabSlopeQ$, $\SelfStabSlopeR$, $\SelfStabSlopeV$ for the quartic, relativistic and VR maps; \cref{fig:selfstab}). This is the averaged force of \cref{sec:connections}. The VR map stalls: beyond its stable range the iterate spends long stretches at large momenta where $\phi$ is small, and the sharpness stays far above the edge ($S=\SelfStabVRFinal$ against $\approx2$).

\begin{figure}[t]
\centering
\includegraphics[width=\linewidth]{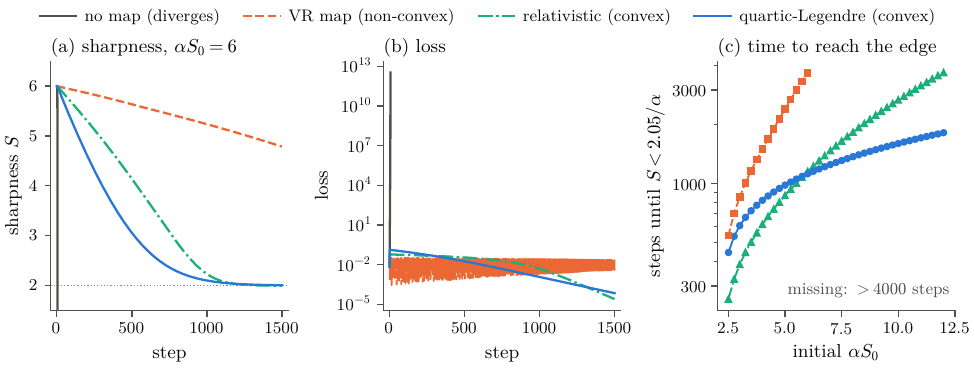}
\caption{Self-stabilization model with $\alpha=1$ (edge at $S=2$), $c=0.3$. (a) Sharpness, (b) loss at $\alpha S_0=6$; (c) steps until $S<2.05$.}
\label{fig:selfstab}
\end{figure}

\section{Delayed loss of stability: details}\label{app:delay}
\paragraph{Controlled ramps.} We ran \RampKickRuns{} ramps. At the slowest rate we also tracked the reference trajectory along $u$, and we exclude the \RampKickExcluded{} runs at $\rho=0.001$ and $R_{\rm in}=0.5$, in which the reference trajectory's own period-two oscillation, driven by training itself, grows to the size of the kick before the kicked deviation returns. At $\rho=0.001$ the largest difference between network and model reaches \RampKickSeriesErrMax\%, near the bottom of the decay just below the edge; in a test run (seed 2, $R_{\rm in}=0.7$) it fell from \RampTestErrA\% to \RampTestErrB\% with a kick of $10^{-6}$ instead of $10^{-4}$, and to \RampTestErrC\% when the sharpness was also tracked at every step instead of every fifth, with the same exit point in all three cases. The sharpness itself drifts by up to \RampLamDrift\% during a ramp. At faster sweeps the exit also depends on the phase of the rotation below the edge at the moment of crossing, so that $R_{\rm out}$ varies non-monotonically with $R_{\rm in}$; the model reproduces this run by run. On an exact linear ramp (thin lines in \cref{fig:ramp}b) the model differs from the network by up to \RampIdealDiff, because the imposed ramp is only approximately linear. For entries closer to the edge than $R_{\rm in}=0.9$ the kick has barely decayed when the edge is crossed, and its transient growth below the edge (by a factor of three to five, since the linear map is not normal) decides the exit. In the runs with noise, the count of exits within the model's 10--90\% range includes runs at $\sigma=10^{-5}$ in which both exit on the first step past the edge. Sharpness and loss after pretraining: \RampLamLo--\RampLamHi{} and \RampLossLo--\RampLossHi.

\paragraph{Reduced model with a kinetic map.} With a map, the exit leads onto the period-two orbit, and $\Reff$ returns to one once the orbit has formed. Propagating the scalar variance closure $s_{t+1}^2=e^{2\Lambda(R_t)}s_t^2+\sigma^2$ of \cref{sec:delay} estimates when the amplitude reaches the orbit. We sweep $R$ linearly from $0.5$ at rate $\rho$ per step in the reduced model with the cubic map and $\beta_1=0.9$, starting from a deviation $10^{-3}c$ with forcing $\sigma$ (\cref{fig:delay}). The peak of $\Reff$ grows with $\rho$. For slow sweeps the deviation decays to the noise floor before the crossing, and the overshoot decreases logarithmically as the noise grows; for fast sweeps it is set by the initial deviation. The linear estimate, taken as the $R$ at which the propagated amplitude reaches half the orbit amplitude, follows the simulation within \DelayPredErrSlow\% for $\rho<0.01$ and overestimates the peak by up to \DelayPredErr\% at the fastest rates.

\paragraph{Sweeps in the networks of \cref{sec:networks}.} For every run at $\mu=0.5$ with plain momentum or a convex map at the absolute scale, we measure the sharpening rate at the crossing (a chord between the two diagnostics around $R=1$) and the largest $\Reff$ while the sharpness climbs (\cref{fig:delay}b). In ResNet-32 the overshoot grows with the sharpening rate (Spearman correlation \DelayResSpearman{} over \DelayResN{} runs, $p=\DelayResSpearmanP$): the \DelayResFastN{} runs that cross at $\rho\ge0.04$ per step peak at $\Reff=\DelayResFastLo$--\DelayResFastHi{}, the \DelayResSlowN{} slower ones at most at \DelayResSlowHi. In the MLP, plain momentum crosses at $\rho\lesssim\DelayMlpMomRhoHi$ and overshoots by at most a factor \DelayMlpMomPeakHi. The reduced model reproduces this trend but not its size: \DelayResAboveBand{} of the \DelayResN{} ResNet peaks lie above the range it gives over initial amplitudes from $10^{-6}$ to $1$. In the MLP a map is already engaged by the large descent momentum before the crossing: with the quartic map, $\Reff$ is \DelayMlpReffCross{} at the crossing (seed 0) and at most \DelayMlpPeakObs{} during the climb to $R=\DelayMlpRmax$.

\begin{figure}[t]
\centering
\includegraphics[width=\linewidth]{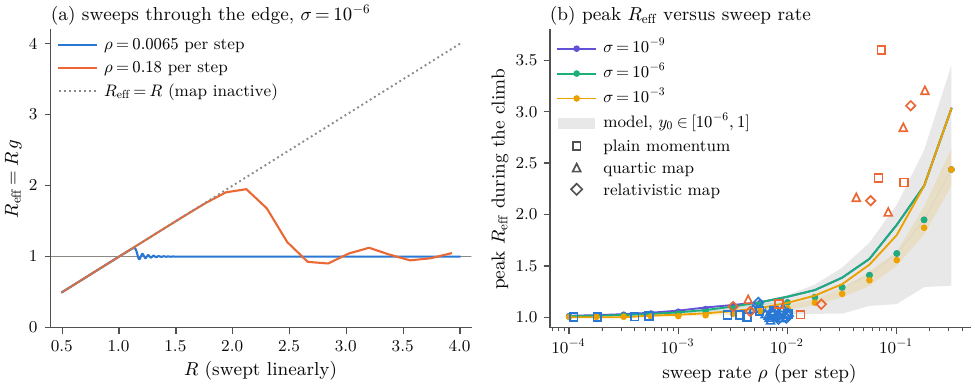}
\caption{Delayed loss of stability, reduced model with the cubic map, $\beta_1=0.9$. (a) $\Reff$ while $R$ is swept linearly through the edge at a slow and a fast rate; the dotted line is $\Reff=R$ (map inactive). (b) Peak $\Reff$ versus sweep rate: simulation (dots, 64 noise realizations; band: 10--90\%) and the linear estimate (lines) for three noise levels; grey band: range over initial amplitudes $10^{-6}$--$1$ and the cubic and quartic maps; open markers: largest $\Reff$ during the climb of the sharpness in the networks ($\mu=0.5$, all seeds; orange ResNet-32, blue MLP).}
\label{fig:delay}
\end{figure}

\section{Networks: additional results}\label{app:networks}
\paragraph{Fashion-MNIST MLP.} With a tanh MLP (784--128--128--10, MSE, 2{,}000 images, one seed), plain momentum diverges for $\mu\ge2$. For $\mu\in\{2,4\}$ every map, the VR map included, trains, and every run ends with $R=\eta\lambda_{\max}/\aeos\in[0.99,1.01]$ (\cref{fig:net}): the network's sharpness self-stabilizes down to the plain edge, where the multiplier returns to one. At $\mu=8$ with a common $c=0.05$ only the VR map reached a low loss; with $c$ chosen per map (\cref{fig:net}c) all maps reach losses between \NetVREight{} and \NetQuartEight{}, the quartic map needing the smallest $c$ (\NetQuartEightC{}). The oscillation amplitude scales with $c$, so in networks $c$ trades the size of the edge oscillation against the large-start risk of \cref{sec:global}.

\begin{figure}[t]
\centering
\includegraphics[width=\linewidth]{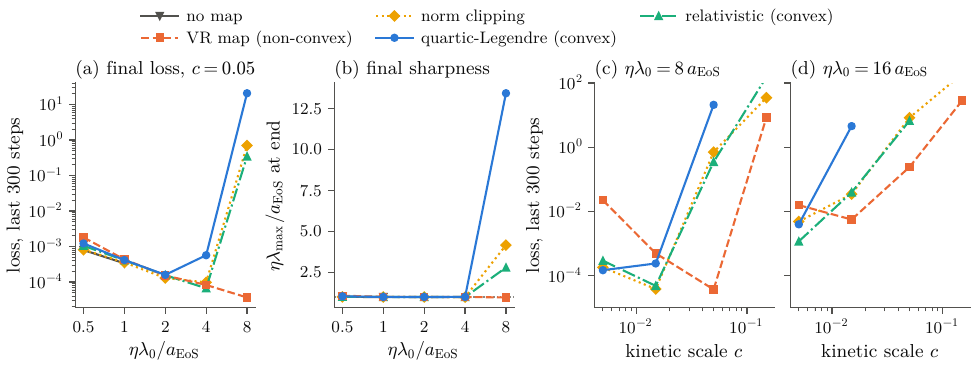}
\caption{Fashion-MNIST MLP: full-batch momentum on a tanh MLP with per-tensor maps (one seed). (a) Final loss and (b) final sharpness versus step size; (c,d) final loss versus $c$ at $8\times$ and $16\times$ the initial edge step.}
\label{fig:net}
\end{figure}

\paragraph{MLP on CIFAR-10: details.} To summarize the transient we take, for every completed run, the median of $\Reff$ over the Lanczos diagnostics with $R\ge1.1$ after step 100; these two thresholds were chosen after a first look at the data. For thresholds $R\ge1.05$, $1.1$, $1.2$ or $1.5$ and steps $\ge50$, $100$ or $200$, the median over runs of this per-run median stays between \ESensCvxLo{} and \ESensCvxHi{} for the convex maps (\ESensCvxNLo--\ESensCvxNHi{} runs contribute) and between \ESensVrLo{} and \ESensVrHi{} for the VR map. The Lanczos estimates converged at all diagnostics, and the two independent starts used for calibration agree to $<10^{-4}$. With the absolute scale no run with a kinetic map diverged (\EMlpGatedAbsDiv{}); \EMlpCvxWithinTen\% of the convex-map runs have a per-run median $\Reff$ within $\pm10\%$ of one. During the climb itself $\Reff$ lags behind ($0.80$ at $R=1.38$ in the run of \cref{fig:e1mlp}a), and pooled over individual diagnostics rather than runs, $\EMlpCvxAbsPtsWithinTen\%$ (absolute scale) and $\EMlpCvxRelPtsWithinTen\%$ (relative scale) of the points lie within $\pm10\%$ of one. The VR map's lowest values come from the relative scale; with the absolute scale and $\mu\le1$ it ends at $R=\EMlpVRSmallRlast$ with a final loss \EMlpVRSmallLossRatio{} times that of plain momentum. The differences between convex maps predicted by the reduced model (relativistic versus quartic) do not appear in the network. For $\mu\ge1$, all completed runs with the quartic or relativistic map and the absolute scale end within $1\%$ of $R=1$ (\EMlpQuarticAbsEndWithin{} and \EMlpRelativisticAbsEndWithin). With the relative scale, \EMlpDivHighQuarticRel{} quartic and \EMlpDivHighRelativisticRel{} relativistic runs with $\mu\ge4$ diverged, and the VR map does not diverge but ends at losses of \EMlpVRRelTopMin--\EMlpVRRelTopMax{} at $\mu=8$, against an initial \EMlpLossZero; here the quartic map fails more often than the relativistic one, the opposite of the reduced model's ranking. Where plain momentum is stable ($\mu\le1$), the maps with the absolute scale only slow training (quartic: \EMlpQSmallLossRatio{} times the final loss of plain momentum).

The cubic map holds the effective step at the edge most tightly (per-run median $\Reff=\EMlpCubicAbsReffMed$, IQR \EMlpCubicAbsReffLo--\EMlpCubicAbsReffHi{}, absolute scale), but it trains worse than the quartic and relativistic maps (\cref{tab:e1mlp}) and, unlike them, does not return to the plain edge: for $\mu\ge1$ it ends at $R=\EMlpCubicAbsEndLo$--\EMlpCubicAbsEndHi{} (within $1\%$ of one: cubic \EMlpCubicAbsEndWithin, quartic \EMlpQuarticAbsEndWithin, relativistic \EMlpRelativisticAbsEndWithin), and at $\mu=8$ its median final loss is \EMlpCubicAbsTopLoss{} against \EMlpQuarticAbsTopLoss{} (quartic) and \EMlpRelativisticAbsTopLoss{} (relativistic). At $\mu=8$ this is consistent with the faster growth of the cubic orbit amplitude (\cref{sec:global}); at $\mu\le4$, where the cubic runs end at $R\approx1.02$--$1.12$, the closed forms give the cubic map the \emph{smallest} amplitude, so amplitude does not explain the gap there. A plausible cause is the shape of the multiplier near zero speed: the cubic map's falls off linearly, $g\approx1-z$, because $|v|^3$ is not analytic at $0$, whereas the quartic, relativistic and VR multipliers fall off quadratically, so any residual momentum keeps the cubic multiplier below one and the sharpness above the plain edge. With the relative scale the cubic map diverged in \EMlpDivHighCubicRel{} runs with $\mu\ge4$.

\begin{table}[t]
\centering\small
\caption{MLP on CIFAR-10: median final training loss over 10 seeds (number of diverged runs in brackets) and per-run median of $\Reff$ over the diagnostics with $R\ge1.1$ after step 100 (median and IQR over runs; number of runs with such points).}
\label{tab:e1mlp}
\resizebox{\linewidth}{!}{\begin{tabular}{lcccccc}
\toprule
& \multicolumn{5}{c}{median final loss (diverged runs)} & $\Reff$ above the edge\\
method & $\mu=0.5$ & $\mu=1$ & $\mu=2$ & $\mu=4$ & $\mu=8$ & median [IQR] (runs)\\
\midrule
momentum (no map) & $1.3{\times}10^{-5}$ & $5.8{\times}10^{-6}$ & $1.6{\times}10^{-6}$ (1) & -- (10) & -- (10) & -- (0)\\
quartic-Legendre, abs. $c$ & $2.4{\times}10^{-4}$ & $2.6{\times}10^{-5}$ & $1.0{\times}10^{-6}$ & $4.7{\times}10^{-9}$ & $3.1{\times}10^{-6}$ & 1.01 [1.00, 1.04] (50)\\
cubic-Legendre, abs. $c$ & $2.2{\times}10^{-4}$ & $1.4{\times}10^{-4}$ & $4.8{\times}10^{-5}$ & $1.8{\times}10^{-5}$ & 0.20 & 1.00 [1.00, 1.01] (50)\\
relativistic, abs. $c$ & $2.3{\times}10^{-4}$ & $1.3{\times}10^{-5}$ & $8.0{\times}10^{-7}$ & $1.5{\times}10^{-10}$ & $1.9{\times}10^{-7}$ & 1.01 [1.00, 1.03] (50)\\
VR map, abs. $c$ & $2.6{\times}10^{-3}$ & $4.9{\times}10^{-3}$ & $3.0{\times}10^{-4}$ & $3.4{\times}10^{-7}$ & $6.3{\times}10^{-7}$ & 0.80 [0.61, 0.91] (50)\\
quartic-Legendre, rel. $c$ & $1.6{\times}10^{-5}$ & $6.6{\times}10^{-6}$ & $2.1{\times}10^{-6}$ & 2.31 (2) & -- (10) & 1.01 [1.00, 1.01] (9)\\
cubic-Legendre, rel. $c$ & $1.7{\times}10^{-5}$ & $7.3{\times}10^{-6}$ & $6.7{\times}10^{-6}$ & 1.75 (9) & -- (10) & 1.10 [1.07, 1.11] (4)\\
relativistic, rel. $c$ & $1.3{\times}10^{-5}$ & $5.7{\times}10^{-6}$ & $3.1{\times}10^{-6}$ & 2.11 & 14.91 (6) & 1.01 [1.00, 1.02] (15)\\
VR map, rel. $c$ & $1.4{\times}10^{-5}$ & $6.7{\times}10^{-6}$ & $4.7{\times}10^{-6}$ & 0.22 & 14.21 & 0.17 [0.12, 0.30] (20)\\
\bottomrule
\end{tabular}
}
\end{table}

\paragraph{ResNet-32: details.} Plain momentum at $\mu\ge2$ that did not diverge catapulted: its sharpness fell to \EResMomCollRminLo--\EResMomCollRminHi{} of the edge, towards a nearly constant predictor (\EResMomHighCollapsed{} runs). All runs with a map at the absolute scale trained with final loss at most \EResAbsHighMaxLoss; at the relative scale \EResRelHighCollapsed{} runs collapsed, \EResRelHighDiv{} diverged and one ended at $0.016$. At $\mu=0.5$ every run trained. Over the few diagnostics above the plain edge, the per-run median of $\Reff$ is \EResCvxMed{} (IQR \EResCvxLo--\EResCvxHi; \EResCvxN{} runs). Neither the convexity nor the elasticity of the map plays a visible role in preventing the collapse (at $\mu=0.5$ the VR and relativistic maps at the absolute scale do end at higher losses, \cref{fig:e1res}c).

\paragraph{Minibatch curvature.} In an earlier minibatch ResNet-32 study (batch 1{,}000, Adam), the preconditioned sharpness sat at $R\approx0.2$--$0.3$ for every learning rate between $10^{-3}$ and $10^{-2}$, i.e.\ it scaled as $1/\eta$ but stayed far below the full-batch threshold. Two explanations are possible: an estimator biased low by the large negative spectral tail of batch-normalized networks (at initialization $\lambda_{\min}\approx-57\lambda_{\max}$, which slows shifted power iteration), or the edge of \emph{stochastic} stability, where a batch sharpness rather than $\lambda_{\max}$ sits at the threshold \citep{andreyev2024eoss}. We repeated the setting (\EFourRuns{} runs: two seeds, $\beta_2\in\{0.95,0.999\}$, four learning rates; \cref{tab:e4}) with Lanczos (full re-orthogonalization; converged when the a-posteriori error estimate $r^2/(\theta_1-\theta_2)$ for the top Ritz value $\theta_1$, with residual $r$ and Ritz gap $\theta_1-\theta_2$, falls below $2\times10^{-3}\theta_1$ and $\theta_1$ has changed by less than $10^{-3}\theta_1$ over the last eight iterations, or when $r<10^{-2}\theta_1$; this criterion was met at every diagnostic, but the Ritz gap only estimates the true spectral gap), the earlier power iteration and the preconditioned batch sharpness on the same batches. Late in training, the preconditioned sharpness of the training batch is $R=\EFourRLo$--\EFourRHi{} by Lanczos, and the earlier estimate is \EFourLegLo--\EFourLegHi{} of it: the low values are real, not an estimator artefact (the negative tail has shrunk to $\lambda_{\min}/\lambda_{\max}=\EFourLamLo$ to \EFourLamHi). The preconditioned batch sharpness, $\lambda_{\rm bs}=\sum_k (P^{-1}\nabla_k)^\top H_k(P^{-1}\nabla_k)\big/\sum_k \nabla_k^\top P^{-1}\nabla_k$ over $K=4$ fresh augmented batches $k$ with gradients $\nabla_k$ and Hessians $H_k$ and the current preconditioner $P=\mathrm{diag}(\sqrt{\hat s}+\epsilon)$ (a ratio of sums), is $\eta\lambda_{\rm bs}=\EFourEtaBsLo$--\EFourEtaBsHi{} ($R_{\rm bs}=\EFourBsLo$--\EFourBsHi) across a tenfold range of learning rates, i.e.\ it scales as $1/\eta$, as at an edge, at a level well below the full-batch adaptive threshold $38/\eta$. This is consistent with dynamics limited by stochastic stability \citep{andreyev2024eoss}, but these experiments do not establish the corresponding stability boundary for Adam. For SGD with heavy-ball momentum, \citet{andreyev2026momentum} find batch-sharpness plateaus below the full-batch threshold; Adam's preconditioned version, which their analysis does not cover, plateaus here at $\eta\lambda_{\rm bs}\approx3$--$4$.

\begin{table}[t]
\centering\small
\caption{Minibatch Adam on ResNet-32 (batch 1{,}000), medians over two seeds of the median over steps $\ge900$. $R$: $\eta\lambda_{\max}/38$ of the preconditioned Hessian of the training batch; legacy/Lanczos: ratio of the earlier power-iteration estimate to the Lanczos estimate on the same batch; $R_{\rm bs}$: preconditioned batch sharpness over $38/\eta$.}
\label{tab:e4}
\begin{tabular}{llccccc}
\toprule
$\beta_2$ & $\eta$ & $R$ (Lanczos) & legacy / Lanczos & $R_{\rm bs}$ & $\eta\lambda_{\rm bs}$ & $\lambda_{\min}/\lambda_{\max}$\\
\midrule
0.95 & 0.001 & 0.20 & 0.96 & 0.086 & 3.3 & -0.08\\
0.95 & 0.002 & 0.24 & 0.92 & 0.106 & 4.0 & -0.07\\
0.95 & 0.005 & 0.23 & 0.91 & 0.106 & 4.0 & -0.10\\
0.95 & 0.01 & 0.27 & 0.96 & 0.096 & 3.6 & -0.10\\
0.999 & 0.001 & 0.17 & 0.94 & 0.076 & 2.9 & -0.09\\
0.999 & 0.002 & 0.20 & 0.93 & 0.088 & 3.3 & -0.08\\
0.999 & 0.005 & 0.22 & 0.97 & 0.101 & 3.8 & -0.08\\
0.999 & 0.01 & 0.24 & 0.97 & 0.095 & 3.6 & -0.09\\
\bottomrule
\end{tabular}

\end{table}

\section{Optimizer benchmarks: additional results}\label{app:practice}
\paragraph{Equal-budget comparisons.} These small models were trained on CPU. Every method receives 10 random-search trials over its own hyperparameters, selected by validation loss, and then three fresh seeds on held-out data (\cref{tab:tuned}; search spaces in \cref{app:search}). For the CNN, validation uses 5{,}000 of the 10{,}000 training images held out from training, and test the first 5{,}000 test images; for the character-level model, both losses are computed on 64 random windows of 64 characters from the same held-out last $10\%$ of the text, so validation and test are not disjoint there. On a Fashion-MNIST CNN, every Adam variant with a fixed scale, VRAdam and the preconditioned controller reach test losses of $0.266$--$0.280$, against AdamW's $0.275\pm0.004$, while the EMA-scale variant and the directional controller are clearly worse; Muon and Muon with a quartic map are best ($0.261$, $0.259$). On a character-level transformer the same Adam variants lie within $1.70$--$1.74$, the EMA-scale variant fails ($2.06$), and Muon ($1.59\pm0.01$) is clearly better than everything else. On an MLP, VRAdam is the only method that still learns at $\eta=0.1$, about $30\times$ AdamW's best learning rate (validation loss $0.91$--$1.01$ over two seeds, against $2.30$ at chance for all others; \cref{fig:lrsens}).

\begin{table}[t]
\centering\small
\caption{Equal-budget comparison: mean $\pm$ s.d.\ over 3 fresh seeds of the configuration selected from 10 trials.}
\label{tab:tuned}
\begin{tabular}{cc}
Fashion-MNIST CNN & char-level transformer\\[2pt]
\resizebox{0.53\linewidth}{!}{\begin{tabular}{lccc}
\toprule
Method & test loss & test acc. & diverged trials \\
\midrule
AdamW & 0.275 $\pm$ 0.004 & 90.15 $\pm$ 0.46 & 0/10 \\
AdamW + clip & 0.280 $\pm$ 0.006 & 90.13 $\pm$ 0.07 & 0/10 \\
VRAdam & 0.269 $\pm$ 0.001 & 90.40 $\pm$ 0.26 & 0/10 \\
K-Adam quartic & 0.270 $\pm$ 0.004 & 90.17 $\pm$ 0.53 & 0/10 \\
K-Adam quartic (EMA $c$) & 0.327 $\pm$ 0.001 & 88.59 $\pm$ 0.08 & 0/10 \\
K-Adam relativistic & 0.267 $\pm$ 0.003 & 90.53 $\pm$ 0.39 & 0/10 \\
K-Adam VR map (abl.) & 0.268 $\pm$ 0.003 & 90.38 $\pm$ 0.07 & 0/10 \\
K-Adam quartic, global (abl.) & 0.266 $\pm$ 0.008 & 90.57 $\pm$ 0.42 & 0/10 \\
EoS controller (precond.) & 0.266 $\pm$ 0.010 & 90.67 $\pm$ 0.49 & 0/10 \\
EoS controller (directional) & 0.315 $\pm$ 0.007 & 88.83 $\pm$ 0.14 & 0/10 \\
Muon & 0.261 $\pm$ 0.007 & 90.95 $\pm$ 0.46 & 0/10 \\
Muon + quartic map & \textbf{0.259 $\pm$ 0.008} & 91.17 $\pm$ 0.54 & 0/10 \\
\bottomrule
\end{tabular}
} &
\resizebox{0.41\linewidth}{!}{\begin{tabular}{lcc}
\toprule
Method & test loss & diverged trials \\
\midrule
AdamW & 1.728 $\pm$ 0.024 & 0/10 \\
AdamW + clip & 1.703 $\pm$ 0.026 & 0/10 \\
VRAdam & 1.708 $\pm$ 0.028 & 0/10 \\
K-Adam quartic & 1.739 $\pm$ 0.027 & 0/10 \\
K-Adam quartic (EMA $c$) & 2.061 $\pm$ 0.001 & 0/10 \\
K-Adam relativistic & 1.719 $\pm$ 0.022 & 0/10 \\
K-Adam VR map (abl.) & 1.722 $\pm$ 0.010 & 0/10 \\
K-Adam quartic, global (abl.) & 1.697 $\pm$ 0.032 & 0/10 \\
EoS controller (precond.) & 1.704 $\pm$ 0.034 & 0/10 \\
EoS controller (directional) & 1.725 $\pm$ 0.039 & 0/10 \\
Muon & \textbf{1.592 $\pm$ 0.011} & 0/10 \\
Muon + quartic map & 1.605 $\pm$ 0.013 & 0/10 \\
\bottomrule
\end{tabular}
}
\end{tabular}
\end{table}

\begin{figure}[t]
\centering
\includegraphics[width=0.75\linewidth]{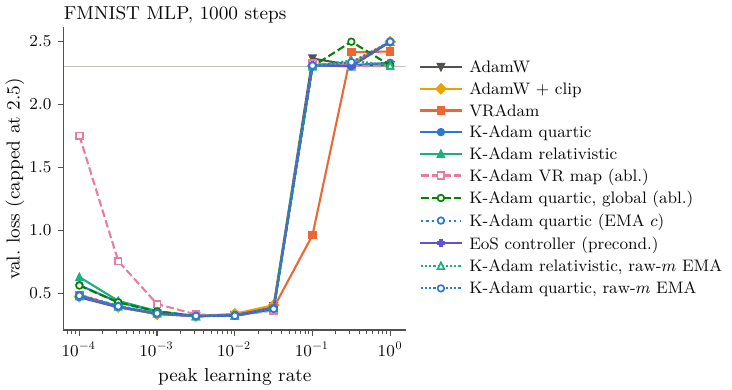}
\caption{Validation loss versus peak learning rate (Fashion-MNIST MLP), with the map hyperparameters fixed in advance.}
\label{fig:lrsens}
\end{figure}

\paragraph{ResNet-18 on CIFAR-10: tuned comparison.} Twelve random-search trials per method (30 epochs, selection by validation cross-entropy on a held-out split of 5k images), then five fresh seeds for 50 epochs evaluated on the test set (batch 512, bf16, cosine schedule, weight decay $5\times10^{-4}$, which the Muon variants apply to all tensors and the Adam variants only to tensors with at least two dimensions; \cref{tab:r18}). Kinetic Muon has the highest mean accuracy and a lower cross-entropy than Muon (paired over seeds: $-0.011\pm0.006$ CE, $+0.17\pm0.11$ pp; the reruns with the same seeds in \cref{tab:e3} give $\EThKinVsMuonCE$ and $\EThKinVsMuonAcc$ pp, a measure of the variation from non-deterministic GPU arithmetic); the quartic Adam variants raise accuracy but also cross-entropy relative to AdamW. The selected configurations differ in effective learning rate: Kinetic Muon's EMA map settles at a multiplier of $0.66$ in these runs, which places it at $1.8\times$ Muon's learning rate for the Muon-updated parameters and $0.08\times$ for the Adam-updated ones. These differences motivated the controls of \cref{tab:e3}, in which the multiplier measured was \EThKinGate{}, slightly above the $0.66$ used to set the controls, so the controls match the effective learning rates only approximately; there, both the quartic map and AdamW at three times its tuned learning rate raise accuracy by about half a point and cross-entropy by $0.02$--$0.03$ relative to tuned AdamW, which was selected for cross-entropy. The two curvature controllers estimate the sharpness with Hessian-vector products; in this study those extra forward passes ran in training mode and therefore also updated the batch-normalization running statistics used at test time. Their test numbers may be affected in either direction, and we report them for completeness only.

\begin{table}[t]
\centering\small
\caption{ResNet-18 / CIFAR-10, five fresh seeds (mean $\pm$ s.d.). $^\dagger$The curvature probes of these two methods also updated the batch-normalization running statistics used at test time (see text); reported for completeness only.}
\label{tab:r18}
\begin{tabular}{lcc}
\toprule
method & test CE & test accuracy (\%)\\
\midrule
AdamW & 0.334 $\pm$ 0.012 & 92.44 $\pm$ 0.14\\
AdamW + clip & 0.334 $\pm$ 0.004 & 92.66 $\pm$ 0.08\\
VRAdam & 0.343 $\pm$ 0.005 & 92.90 $\pm$ 0.23\\
K-Adam quartic & 0.344 $\pm$ 0.009 & 92.78 $\pm$ 0.22\\
K-Adam quartic, global & 0.367 $\pm$ 0.008 & 92.85 $\pm$ 0.12\\
K-Adam relativistic & 0.340 $\pm$ 0.006 & 92.35 $\pm$ 0.11\\
K-Adam VR map (per tensor) & 0.359 $\pm$ 0.014 & 92.91 $\pm$ 0.24\\
K-Adam quartic, EMA scale & 0.347 $\pm$ 0.006 & 91.58 $\pm$ 0.22\\
EoS controller (precond.)$^\dagger$ & 0.335 $\pm$ 0.011 & 92.39 $\pm$ 0.15\\
EoS controller (directional)$^\dagger$ & 0.382 $\pm$ 0.013 & 90.48 $\pm$ 0.21\\
Muon & 0.353 $\pm$ 0.005 & 93.99 $\pm$ 0.15\\
Muon + quartic EMA map & 0.342 $\pm$ 0.006 & 94.16 $\pm$ 0.24\\
\bottomrule
\end{tabular}
\end{table}

\paragraph{Small GPT.} A \EFiveParams M-parameter GPT-2-style model (8 layers, width 512, context 1{,}024) trained for 6{,}000 steps of $64\times1{,}024$ tokens on FineWeb-Edu \citep{penedo2024fineweb}, with AdamW, AdamW with gradient clipping at $1.0$, VRAdam, and AdamW with the cubic map on $\hat m$ at two scales, over learning rates $10^{-3}$--$3\times10^{-2}$ and two seeds (\EFiveRuns{} runs, \EFiveTokens M tokens each, about \EFiveMinutes{} minutes per run on one A100; \cref{tab:e5,fig:gpt}). The data are the first file of the sample-10BT subset, tokenized with the GPT-2 tokenizer; its first $10$M tokens (whole documents) are held out for validation and the following tokens are used for training. For the maps, $c_B=k\,\mathrm{rms}(\hat m_B)$ is frozen at the end of the 300-step warm-up. The data order is the same in every run, and the two seeds change only the initialization, which is shared across methods. Every method is best at $\eta=0.003$, where the mean final validation losses span only \EFiveBestSpread{} nats (AdamW \EFiveAdamwBest, clipping \EFiveClipBest, VRAdam \EFiveVrBest, cubic map \EFiveKinFourBest{} with $k=4$ and \EFiveKinOneBest{} with $k=1$), against a spread over initialization seeds of \EFiveAdamwSeedSpread{} for AdamW. The maps act as learning-rate reductions there (mean multiplier \EFiveKinOneGateThree{} for $k=1$ and \EFiveKinFourGateThree{} for $k=4$; VRAdam's factor is \EFiveVrGateThree). At $\eta=0.01$ AdamW fails (\EFiveAdamwTen), clipping largely avoids the failure (\EFiveClipTen, against its best of \EFiveClipBest), the cubic map with $k=1$ partly does, with a large seed dependence (\EFiveKinOneTen), and VRAdam does not (\EFiveVrTen), although its multiplier falls to \EFiveVrGateTen; at $\eta=0.03$ every method fails.

\begin{table}[t]
\centering\small
\caption{Small GPT (\EFiveParams M parameters, \EFiveTokens M tokens of FineWeb-Edu): final validation loss of both seeds (sorted) at each peak learning rate.}
\label{tab:e5}
\begin{tabular}{lcccc}
\toprule
& \multicolumn{4}{c}{final validation loss (two seeds) at peak learning rate}\\
method & $0.001$ & $0.003$ & $0.01$ & $0.03$\\
\midrule
AdamW & 3.77 / 3.78 & 3.70 / 3.73 & 6.47 / 6.69 & 6.68 / 6.71\\
AdamW + clip 1.0 & 3.78 / 3.78 & 3.71 / 3.72 & 3.79 / 4.07 & 6.74 / 6.82\\
VRAdam & 3.78 / 3.78 & 3.70 / 3.71 & 5.99 / 6.09 & 6.57 / 6.59\\
cubic map on $\hat m$, $k=1$ & 3.87 / 3.88 & 3.73 / 3.74 & 3.96 / 4.84 & 5.24 / 5.86\\
cubic map on $\hat m$, $k=4$ & 3.81 / 3.81 & 3.70 / 3.71 & 5.05 / 5.71 & 6.13 / 6.42\\
\bottomrule
\end{tabular}

\end{table}

\begin{figure}[t]
\centering
\includegraphics[width=\linewidth]{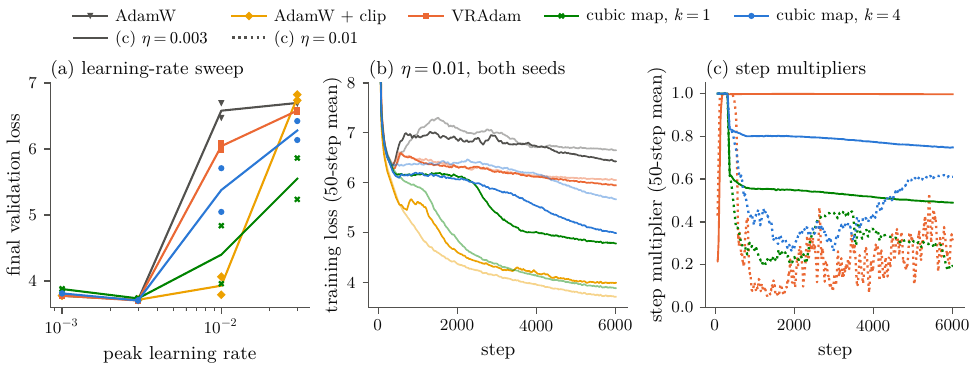}
\caption{Small GPT. (a) Final validation loss versus peak learning rate (lines: mean of two seeds). (b) Training loss at $\eta=0.01$. (c) Step multipliers of VRAdam and the cubic map at $\eta=0.003$ (solid) and $0.01$ (dotted), seed 0.}
\label{fig:gpt}
\end{figure}

\section{Grokking test}\label{app:grok}
The question is whether the loss spikes that accompany grokking with adaptive optimizers \citep{thilak2022slingshot} are needed for it. We trained the one-layer transformer of \citet{nanda2023progress} on modular addition ($p=113$, $30\%$ of pairs, full batch, 50{,}000 steps) with AdamW (weight decay 1) and with Adam without weight decay, using the five optimizers of the GPT study and three seeds each (\cref{tab:e6,fig:grok}). With weight decay all \ESixWdOneN{} runs grokked, after \ESixGrokLo k--\ESixGrokHi k steps. Clipping and the cubic map with $k=1$ lowered the median of the largest training loss after memorization about fourfold (\ESixClipSpikeOne{} and \ESixKinOneSpikeOne{} against \ESixAdamwSpikeOne{} for AdamW; not on every seed) without a significant delay: the median steps to grok, \ESixClipGrok k and \ESixKinOneGrok k against \ESixAdamwGrok k, lie within AdamW's own seed range, and across runs the spike size shows no significant correlation with the time to grok (Spearman \ESixSpearman, $p=\ESixSpearmanP$). This spike measure, chosen after looking at the curves, extends to the end of training and so includes the time after grokking; restricted to the window before grokking, the correlation is positive (Spearman \ESixSpearmanPre, $p=\ESixSpearmanPreP$, \ESixPreN{} runs): larger spikes went with later, not earlier, grokking, although part of this is mechanical, since a longer window tends to contain a larger maximum. The cubic map with $k=4$ grokked earlier than AdamW on each of the three seeds (by 1.4k--2.4k steps), a difference that three seeds cannot establish. VRAdam's damper reduced the spikes inconsistently. Without weight decay no run grokked (test accuracy at most \ESixWdZeroAccHi), and there the maps produced the largest spikes of the study on some seeds (\cref{tab:e6}). In this setting, reducing the measured loss peaks did not consistently delay grokking, and the maps are not a reliable way to remove the peaks.

\begin{table}[t]
\centering\small
\caption{Grokking test (modular addition, $p=113$): steps until test accuracy $\ge0.95$, final test accuracy (median), and the largest training loss from 1{,}000 steps after memorization to the end, for each seed.}
\label{tab:e6}
\resizebox{\linewidth}{!}{\begin{tabular}{llccc}
\toprule
regime & method & steps to grok (seeds 0, 1, 2) & final test acc. & largest loss after memorisation (seeds 0, 1, 2)\\
\midrule
AdamW, wd 1 & AdamW & 20.4k, 14.6k, 14.5k & 1.00 & 0.71, 0.21, 0.57\\
AdamW, wd 1 & AdamW + clip 1.0 & 21.9k, 14.2k, 18.4k & 1.00 & 0.14, 0.04, 0.15\\
AdamW, wd 1 & VRAdam & 19.9k, 12.7k, 15.5k & 1.00 & 0.45, 0.05, 0.61\\
AdamW, wd 1 & cubic map on $\hat m$, $k=1$ & 21.1k, 11.4k, 15.6k & 1.00 & 0.04, 0.23, 0.13\\
AdamW, wd 1 & cubic map on $\hat m$, $k=4$ & 18.0k, 12.6k, 13.1k & 1.00 & 0.37, 0.02, 0.33\\
\midrule
Adam, wd 0 & AdamW & --, --, -- & 0.26 & 2.62, 5.43, 4.16\\
Adam, wd 0 & AdamW + clip 1.0 & --, --, -- & 0.25 & 0.24, 1.62, 0.15\\
Adam, wd 0 & VRAdam & --, --, -- & 0.26 & 1.72, 22, 4.77\\
Adam, wd 0 & cubic map on $\hat m$, $k=1$ & --, --, -- & 0.25 & 0.10, 546, 0.12\\
Adam, wd 0 & cubic map on $\hat m$, $k=4$ & --, --, -- & 0.26 & 190, 6.47, 58\\
\bottomrule
\end{tabular}
}
\end{table}

\begin{figure}[t]
\centering
\includegraphics[width=\linewidth]{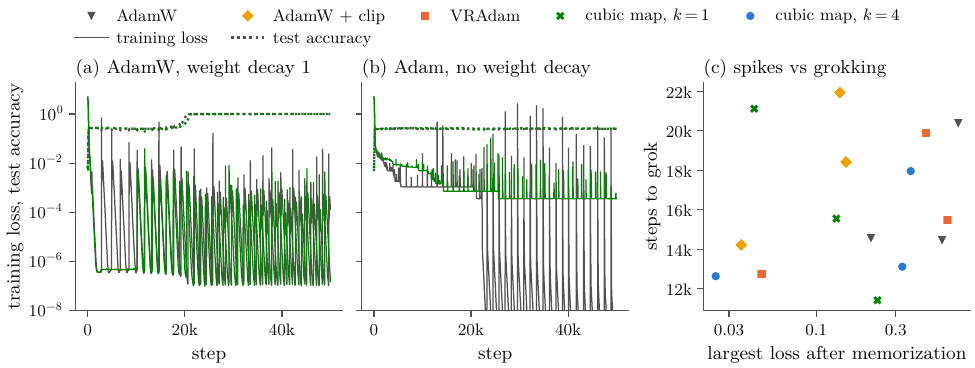}
\caption{Grokking test. (a,b) Training loss (solid) and test accuracy (dotted) of AdamW and the cubic map ($k=1$), seed 0, with and without weight decay. (c) Steps to grok versus the largest training loss after memorization, weight decay 1, all methods and seeds.}
\label{fig:grok}
\end{figure}

\FloatBarrier
\section{Notation, protocols and implementation details}\label{app:search}
\paragraph{Notation.} \Cref{tab:notation} collects the symbols used in more than one section and the conventions for the scale $c$ of the map, which differ between studies.
\begin{table}[H]
\centering\small
\caption{Notation.}\label{tab:notation}
\resizebox{\linewidth}{!}{\begin{tabular}{@{}ll@{}}
\toprule
Symbol & Meaning \\
\midrule
$\eta,\ \eta_t$ & learning rate \\
$\lambda,\ \lambda_{\max}$ & curvature (Hessian eigenvalue; for Adam, of $P^{-1/2}HP^{-1/2}$) and its largest value \\
$b=\beta_1$ & EMA momentum coefficient; $\kappa=(1-b)/(1+b)$, $\aeos=2/\kappa$ \\
$a=\eta\lambda,\ A=\kappa a$ & step size of the reduced model; $A=2$ at the edge \\
$R,\ \Reff$ & sharpness ratio $\eta\lambda_{\max}/\aeos$; effective ratio $R\,\bar g$ (\cref{eq:reffnet}) \\
$\mu$ & value of $R$ at initialization, $\eta=\mu\,\aeos/\lambda_0$ (\cref{sec:quadratic,sec:fullbatch}) \\
$q$ & quantity the map acts on: the momentum, or Adam's $\hat m/(\sqrt{\hat s}+\epsilon)$ \\
$\hat m,\ \hat s$ & Adam's bias-corrected first and second moments (VRAdam writes the first moment $v_t$) \\
$y,\ v$ & gradient and momentum of the reduced model \eqref{eq:reduced}, in units of $c$ \\
$\nabla_B,\ \omega$ & gradient of tensor $B$ and weight decay in \cref{alg:kadam} \\
$z,\ \phi,\ g$ & normalized momentum magnitude $\|q\|/c$; speed profile $\phi(z)=z\,g(z)$; step multiplier $g$ \\
$g_B,\ \bar g$ & multiplier of tensor $B$; mode-weighted multiplier $\sum_B g_B\|u_B\|^2$ \\
$\varepsilon,\ \varepsilon^\ast,\ p$ & elasticity $z\phi'/\phi$; its value on the orbit; its limit as $z\to\infty$ \\
$s,\ n$ & speed of the Legendre power map and its exponent, $s+s^n=z$, $p=1/n$ \\
$\nu$ & eigenvalues of the one-step Jacobian on the orbit \\
$\Lambda,\ \rho$ & local growth rate per step $\log r(J_0)$; sweep rate of $R$ per step \\
$x_t,\ d_t$ & deviation of two copies along the top eigenvector and its momentum component \\
$\sigma$ & noise level: forcing of the reduced state in \eqref{eq:noisyexit}; per-coordinate gradient noise in the network ramps \\
\midrule
\multicolumn{2}{@{}l@{}}{Argument and scale $c$ of the map, by study} \\
\cref{sec:quadratic} & Euclidean norm of the momentum; $c=0.05\,\|\nabla L(\theta_0)\|$ \\
\cref{sec:fullbatch}, absolute & Euclidean norm of the momentum of each tensor; $c_B=0.05$ \\
\cref{sec:fullbatch}, relative & as above, with $c_B=2\|\nabla_BL(\theta_0)\|$ \\
\cref{sec:adam}, map on $q$ & RMS of $q_B$; $c=0.5$ \\
\cref{sec:adam}, map on $\hat m$ & Euclidean norm of $\hat m_B$; $c=0.05$ \\
\cref{alg:kadam}, fixed & RMS of $q_B$ (or $\hat m_B$) divided by a constant $c$ \\
\cref{alg:kadam}, EMA & as above, with $c_B=c\,\big(\text{bias-corrected EMA of }\mathrm{rms}^2\big)^{1/2}$, decay $0.99$ \\
Kinetic Muon & RMS of the Muon update $\Delta_B$, with the EMA scale as above \\
GPT and grokking & RMS of $\hat m_B$; $c_B=k\,\mathrm{rms}(\hat m_B)$, frozen at the end of the warm-up \\
\bottomrule
\end{tabular}
}
\end{table}

\paragraph{Protocols.} \Cref{tab:protocol} lists, study by study, the data and splits, model and loss, optimizer settings, schedule, batch, training length, precision, the scale of the map, seeds and the failure criterion. Every entry was checked against the code, configuration and result files.
{\small\begin{longtable}{@{}p{0.17\linewidth}p{0.80\linewidth}@{}}
\caption{Study protocols.}\label{tab:protocol}\\
\toprule
\raggedright Study & \raggedright Protocol \tabularnewline
\midrule
\raggedright\emph{Conventions} &
\raggedright $\mu=\eta\lambda_0/a_{\mathrm{EoS}}$ with $a_{\mathrm{EoS}}=38$ ($\beta_1=0.9$); EMA momentum $m\leftarrow\beta_1m+(1-\beta_1)\nabla L$ without bias correction; $q=\hat m/(\sqrt{\hat s}+\epsilon)$; per tensor: one scale $c_B$ and one multiplier $g(r_B)$ per parameter tensor, $r_B=\text{argument}/c_B$; $\omega$: decoupled weight decay; $\ge$2-D: weight matrices, convolution kernels, embeddings; none: verified to be absent.
\tabularnewline\midrule
\raggedright Exact quadratic (\cref{sec:quadratic}) &
\raggedright \emph{Data:} first 2{,}000 Fashion-MNIST training images (standardized pixels plus a bias column), one-hot targets; no validation or test set. \emph{Model, loss:} linear least squares, $W\in\mathbb{R}^{785\times10}$ (7{,}850 parameters, one tensor), $L=\tfrac1{2N}\|XW-Y\|_F^2$, $\lambda_{\max}=300.4$. \emph{Optimizer:} full-batch EMA momentum, $\beta_1=0.9$, $\omega=0$; no map, quartic, cubic, relativistic, VR. \emph{LR:} constant, $\eta=\mu\cdot38/\lambda_{\max}$, 16 values $\mu\in[0.5,12]$. \emph{Batch:} full. \emph{Length:} 6{,}000 steps. \emph{Precision:} FP64 (NumPy, integrated in the eigenbasis of $C$). \emph{Scale:} $c=0.05\,\|\nabla L(W_0)\|=0.542$ (initial gradient norm of the seed-0 start, reused for seeds 1, 2); one global block; argument $\|m\|$. \emph{Seeds:} 0, 1, 2 (initialization $W_0\sim\mathcal N(0,0.01^2)$); 240 runs. \emph{Failure:} non-finite loss or excess loss $>10^{12}$.
\tabularnewline
\raggedright Controlled ramps (\cref{sec:delay,app:delay}) &
\raggedright \emph{Data:} 2{,}000 Fashion-MNIST training images (fixed random subset, permutation seed 0), one-hot targets; no validation or test set. \emph{Model, loss:} tanh MLP 784--128--128--10, $L=\tfrac1{2N}\sum_i\|f(x_i)-e_{y_i}\|^2$. \emph{Optimizer:} full-batch EMA momentum, $\beta_1=0.9$, $\omega=0$, no map; reference and perturbed copy take the same $\eta_t$. \emph{LR:} pretraining 1{,}500 steps at $R=0.5$ ($\eta$ reset from $\lambda_{\max}$ every 50 steps); ramp $\eta_t=(R_{\rm in}+\rho t)\,38/\hat\lambda_t$, power iteration warm-started from the end of pretraining (200 iterations), then 4 iterations every 5 steps. \emph{Grid:} kick $\delta=10^{-6}$, $\rho\in\{0.003,0.01,0.03,0.1\}$, $R_{\rm in}\in\{0.1,0.3,0.5,0.7,0.9,0.97\}$; $\delta=10^{-4}$ at $\rho=0.001$, $R_{\rm in}\in\{0.5,0.7,0.9,0.97\}$; noise only ($R_{\rm in}=0.5$, $\sigma\in\{10^{-9},10^{-7},10^{-5}\}$, five rates $\rho=0.001$ to $0.1$, exit at $|x_t|=10^{-4}$); kick $10^{-6}$ plus noise ($\rho=0.003$, $R_{\rm in}\le0.9$); 174 runs. \emph{Batch:} full. \emph{Length:} up to $R=2$ ($R=4$ for $\rho\ge0.05$), stopped 10 steps after $|x_t|$ exceeds $100\times$ the exit threshold. \emph{Precision:} FP64 (CPU). \emph{Scale:} no map. \emph{Seeds:} 0, 1, 2 (initialization; noise generator seed $1000+{}$seed). \emph{Failure:} none defined. \emph{Tracking test} (\cref{app:delay}): seed 2, $R_{\rm in}=0.7$, $\rho=0.001$, kick $10^{-4}$ or $10^{-6}$, sharpness re-estimated every fifth step or at every step; 3 runs.
\tabularnewline
\raggedright Full-batch CIFAR-10: MLP and ResNet-32 (E1; \cref{sec:networks,app:networks}) &
\raggedright \emph{Data:} balanced 1{,}000-image subset (first 100 per class of the 45k training split; split by permutation seed 0 of the 50k images), normalized with subset statistics; 2{,}000 images of the 5k validation split for diagnostics only; no test set. \emph{Model:} tanh MLP 3072--128--128--10 (411{,}146 parameters); ResNet-32 ($3\times5$ basic blocks, widths 16/32/64; 466{,}906 parameters), BN without running statistics; loss and gradients summed over 4 fixed blocks of 250 images (BN statistics per block). \emph{Loss:} MLP $\tfrac1{2N}\sum_i\|f(x_i)-e_{y_i}\|^2$; ResNet mean MSE, $\tfrac1{10N}\sum_i\|f(x_i)-e_{y_i}\|^2$. \emph{Optimizer:} full-batch EMA momentum, $\beta_1=0.9$, $\omega=0$; no map, quartic, cubic, relativistic, VR. \emph{LR:} constant, $\eta=\mu\cdot38/\lambda_0$, $\lambda_0$ by Lanczos at initialization for each seed (2 starts; MLP 72--93, ResNet 27--47); MLP $\mu\in\{0.5,1,2,4,8\}$, ResNet $\mu\in\{0.5,2,4\}$. \emph{Batch:} full. \emph{Length:} 3{,}000 steps. \emph{Precision:} FP32, TF32 off, deterministic cuDNN (A100). \emph{Scale:} per tensor, argument $\|m_B\|$; absolute $c=0.05$ or relative $c_B=2\|\nabla_BL(\theta_0)\|$. \emph{Seeds:} MLP 0--9, ResNet 0--4 (initialization; data fixed); 450 and 135 runs. \emph{Failure:} loss non-finite or $>10^3$, non-finite gradient, or parameter norm non-finite or $>10^6$.
\tabularnewline
\raggedright Full-batch Adam (E2; \cref{sec:adam}) &
\raggedright \emph{Data, model, loss, batch, precision, failure:} as for ResNet-32 in E1. \emph{Optimizer:} Adam, $\beta=(0.9,0.999)$, $\epsilon=10^{-7}$, bias-corrected, $\omega=0$; VRAdam ($\beta_3=1$, $\alpha_1=19$, global $\|m\|^2$ of the uncorrected first moment, same $\beta,\epsilon$); Adam with a per-tensor relativistic map. \emph{LR:} constant, $\eta\in\{0.002,0.01\}$. \emph{Length:} 3{,}000 steps. \emph{Scale:} map on $q$: argument $\mathrm{rms}(q_B)$, $c=0.5$; map on $\hat m$: argument $\|\hat m_B\|$ (Euclidean norm), absolute $c=0.05$; both multiply $q_B$. \emph{Seeds:} 0, 1, 2; 24 runs.
\tabularnewline
\raggedright Fashion-MNIST MLP (\cref{app:networks}) &
\raggedright \emph{Data, model, loss:} as for the controlled ramps. \emph{Optimizer:} full-batch EMA momentum, $\beta_1=0.9$, $\omega=0$; per-tensor none, VR, quartic, relativistic and clipping maps; global VR and quartic maps. \emph{LR:} constant, $\eta=\mu\cdot38/\lambda_0$, $\lambda_0=33.9$ (power iteration at initialization), $\mu\in\{0.5,1,2,4,8\}$; scale sweep at $\mu\in\{8,16\}$. \emph{Batch:} full. \emph{Length:} 3{,}000 steps. \emph{Precision:} FP32 (CPU). \emph{Scale:} absolute $c=0.05$, argument $\|m_B\|$ per tensor (or $\|m\|$ for the global maps); sweep $c\in\{0.005,0.015,0.05,0.15\}$ (per tensor; VR, relativistic, quartic, clipping). \emph{Seeds:} 0 only; $35+32$ runs. \emph{Failure:} loss non-finite or $>10^3$.
\tabularnewline
\raggedright Minibatch curvature (E4; \cref{app:networks}) &
\raggedright \emph{Data:} CIFAR-10 45k training split, per-channel normalization with 50k statistics, random crop (reflection padding 4) and flip; no evaluation set. \emph{Model, loss:} ResNet-32, BN without running statistics; mean MSE. \emph{Optimizer:} Adam, $\beta_1=0.9$, $\beta_2\in\{0.95,0.999\}$, $\epsilon=10^{-8}$, $\omega=0$, no map. \emph{LR:} constant, $\eta\in\{0.001,0.002,0.005,0.01\}$. \emph{Batch:} 1{,}000, reshuffled every epoch. \emph{Length:} 1{,}800 steps (40 epochs); curvature every 100 steps (Lanczos, at most 64 iterations; earlier power iteration, 5 warm-started iterations; batch sharpness over $K=4$ fresh batches). \emph{Precision:} training in bf16 autocast; probes in FP32, TF32 off, deterministic cuDNN (A100). \emph{Scale:} no map. \emph{Seeds:} 0, 1; 16 runs. \emph{Failure:} non-finite loss. \emph{Earlier study cited in the text:} the same data, model, loss, batch, constant learning rates and length, with six methods (AdamW, VRAdam, three per-tensor maps and the preconditioned controller), $\beta_2=0.95$, seed 0, and the sharpness every 50 steps by warm-started shifted power iteration; per-run summaries only.
\tabularnewline
\raggedright Equal-budget CNN and character-level transformer (\cref{sec:practice,app:practice}) &
\raggedright \emph{Data:} Fashion-MNIST: training = first 50k training images, validation = last 10k (first 5k evaluated), test = first 5k of the 10k test images; character task: Shakespeare text (\texttt{input.txt}, 1{,}115{,}394 characters, 65 symbols), first 90\% for training, validation and test each 64 random 64-character windows (generator seeds 123 and 456) from the last 10\%. \emph{Model:} CNN (two $3\times3$ convolutions, 32 and 64 channels, ReLU, max-pooling; FC 128; 10 outputs); pre-LN transformer (2 layers, $d=128$, 4 heads, context 64). \emph{Loss:} CE. \emph{Optimizer:} 12 methods (\cref{app:search}); Adam variants and VRAdam $\beta=(0.9,0.999)$, $\epsilon=10^{-8}$; Muon momentum 0.95 (Nesterov, 5 Newton--Schulz steps), with AdamW, $\beta=(0.9,0.95)$, for 1-D tensors and the transformer's embeddings and head; $\omega=0$ for all methods. \emph{LR:} linear warm-up 50 steps, cosine decay to $0.1\times$ peak. \emph{Batch:} 128 images or 32 sequences of 64 characters, drawn with replacement at every step. \emph{Length:} 800 steps. \emph{Precision:} FP32 (CPU). \emph{Scale:} K-Adam: argument $\mathrm{rms}(q_B)$ per tensor (or over all parameters for the global variant), fixed $c$; EMA variant $c_B=c\,(\text{bias-corrected EMA}_{0.99}\text{ of }\mathrm{rms}^2(q_B))^{1/2}$; Muon map in update space with the EMA scale. \emph{Seeds:} 10 random-search trials on seed 0, selected by validation loss, then seeds 1, 2, 3 on the test split (the seed sets initialization and batch order); 156 runs per task. \emph{Failure:} loss non-finite or $>10^4$.
\tabularnewline
\raggedright MLP learning-rate sensitivity (\cref{fig:lrsens}) &
\raggedright \emph{Data:} Fashion-MNIST as above; validation loss on 5k images. \emph{Model, loss:} ReLU MLP 784--256--256--10; CE. \emph{Optimizer:} as in the equal-budget row, 11 methods with hyperparameters fixed in advance (VRAdam $\beta_3=1$, $\alpha_1=19$; clipping at 1.0; preconditioned controller $\sigma=0.8$). \emph{LR:} 9 log-spaced peak values from $10^{-4}$ to 1; warm-up 50 steps, cosine decay to $0.1\times$ peak. \emph{Batch:} 128. \emph{Length:} 1{,}000 steps. \emph{Precision:} FP32 (CPU). \emph{Scale:} fixed $c=0.15$, argument $\mathrm{rms}(q_B)$ (per tensor or global); EMA variant on $q$, $c=1.5$; two EMA variants on $\hat m$, $c=2.0$, argument $\mathrm{rms}(\hat m_B)$. \emph{Seeds:} 0, 1; 198 runs. \emph{Failure:} loss non-finite or $>10^4$.
\tabularnewline
\raggedright ResNet-18, tuned comparison (\cref{app:practice}, \cref{tab:r18}) &
\raggedright \emph{Data:} CIFAR-10, 45k training / 5k validation (permutation seed 0) / 10k test images; per-channel normalization with 50k statistics; random crop (reflection padding 4) and flip. \emph{Model, loss:} CIFAR ResNet-18 ($3\times3$ stem, BN with running statistics); CE. \emph{Optimizer:} 12 methods (\cref{app:search}); $\beta=(0.9,0.95)$, $\epsilon=10^{-8}$; $\omega=5\times10^{-4}$ on $\ge$2-D tensors (Muon variants: on all tensors); Muon momentum 0.95 (Nesterov). \emph{LR:} linear warm-up 2 epochs, cosine decay to 0. \emph{Batch:} 512. \emph{Length:} search 30 epochs (2{,}610 steps), final runs 50 epochs (4{,}350 steps). \emph{Precision:} bf16 autocast, cuDNN benchmark; controller Hessian-vector products in FP32, in training mode. \emph{Scale:} as in the equal-budget row. \emph{Seeds:} 12 random-search trials on seed 0, selected by validation CE, then seeds 1--5 on the test set. \emph{Failure:} non-finite loss. \emph{Run records:} per-method means and standard deviations, selected hyperparameters and paired differences as reported from the cluster runs; per-seed records are not included.
\tabularnewline
\raggedright ResNet-18, effective-learning-rate controls (E3; \cref{sec:practice}, \cref{tab:e3}) &
\raggedright \emph{Protocol:} final runs of the row above (50 epochs, 4{,}350 steps), re-implemented with multiplier logging. \emph{Configurations} ($\eta$; $\eta_{\rm Adam}$; $c$): Muon, tuned (0.0113; $3.69\times10^{-3}$); Kinetic Muon, tuned (0.0313; $4.55\times10^{-4}$; EMA scale $c=0.92$); Muon at $0.66\times$ Kinetic Muon's rates (0.0207; $3.01\times10^{-4}$) and at its nominal rates; AdamW, tuned ($4.08\times10^{-4}$) and at $3\times$ ($1.22\times10^{-3}$); K-Adam quartic, global ($1.61\times10^{-3}$; fixed $c=0.527$, argument $\mathrm{rms}(q)$ over all parameters). \emph{Precision:} bf16 autocast, cuDNN benchmark (A100). \emph{Seeds:} 1--5, shared by all configurations (paired); 35 runs. \emph{Failure:} non-finite loss (none occurred).
\tabularnewline
\raggedright Small GPT (E5; \cref{sec:practice,app:practice}) &
\raggedright \emph{Data:} FineWeb-Edu \texttt{sample-10BT}, first file only (\texttt{sample/10BT/000\_00000.parquet}; 444{,}978 documents used); GPT-2 BPE (\texttt{tiktoken} \texttt{gpt2}, end-of-text token after each document; vocabulary 50{,}257, embedding padded to 50{,}304). \emph{Validation:} the first 10.0M tokens of that file (whole documents); the next 450.0M tokens form the training set, so no document is shared (no separate shard); final loss on 80 batches of $64\times1{,}024$ tokens (5.24M) at evenly spaced offsets, intermediate loss on 16 batches every 500 steps. \emph{Model:} GPT-2 style, 8 layers, $d=512$, 8 heads, context 1{,}024, tied embeddings, no biases in linear layers; 51{,}463{,}168 parameters. \emph{Loss:} next-token CE. \emph{Optimizer:} AdamW, $\beta=(0.9,0.95)$, $\epsilon=10^{-8}$, $\omega=0.1$ on $\ge$2-D tensors; the same with global-norm clipping at 1.0; VRAdam ($\beta_3=1$, $\alpha_1=19$); AdamW with the cubic map. \emph{LR:} peak $\eta\in\{10^{-3},3\times10^{-3},10^{-2},3\times10^{-2}\}$; linear warm-up 300 steps, cosine decay to $0.1\times$ peak. \emph{Batch:} 64 sequences of 1{,}024 tokens, i.e.\ 65{,}536 tokens per step (4 micro-batches of 16), random windows. \emph{Length:} 6{,}000 steps, 393.2M tokens processed. \emph{Precision:} bf16 autocast, TF32 enabled, \texttt{torch.compile}; one A100 (40\,GB). \emph{Scale:} per tensor, $c_B=k\,\mathrm{rms}(\hat m_B)$ frozen at the end of warm-up, $k\in\{1,4\}$; argument $\mathrm{rms}(\hat m_B)$, multiplier on $q_B$; no map during warm-up. \emph{Seeds:} 0, 1 (initialization only; the data order, seed 1234, is the same for all runs); 40 runs. \emph{Failure:} non-finite loss, or mean loss over 50 steps $>9$ after step 1{,}000 (not triggered).
\tabularnewline
\raggedright Grokking test (E6; \cref{app:grok}) &
\raggedright \emph{Data:} modular addition $a+b \bmod 113$, input $(a,b,{=})$; 30\% of the 12{,}769 pairs (3{,}830) for training, the other 8{,}939 for testing; split drawn per seed. \emph{Model:} one-layer transformer ($d=128$, 4 heads, ReLU MLP of width 512, no LayerNorm, learned positions, untied unembedding). \emph{Loss:} CE at the last position. \emph{Optimizer:} AdamW, $\beta=(0.9,0.98)$, $\epsilon=10^{-8}$, $\omega=1$ on $\ge$2-D tensors (regime wd1) or $\omega=0$ (wd0); the same with clipping at 1.0; VRAdam ($\beta_3=1$, $\alpha_1=19$); cubic map. \emph{LR:} $10^{-3}$, linear warm-up 10 steps, then constant. \emph{Batch:} full. \emph{Length:} 50{,}000 steps (logged every 20). \emph{Precision:} FP32, no autocast (A100). \emph{Scale:} as in E5, frozen after the 10-step warm-up, $k\in\{1,4\}$. \emph{Seeds:} 0, 1, 2 (initialization and split); 30 runs. \emph{Failure:} none defined (grokked: test accuracy $\ge0.95$).
\tabularnewline
\bottomrule

\end{longtable}}

\paragraph{Search spaces.}
\begin{center}\small
\resizebox{\linewidth}{!}{\begin{tabular}{lp{5.2cm}p{5.2cm}}
\toprule
Method & CNN, char-level transformer & ResNet-18 \\
\midrule
AdamW & $\eta$: log [0.0001, 0.3] & $\eta$: log [0.0001, 0.01] \\
AdamW + clip & $\eta$: log [0.0001, 0.3]; clip: log [0.1, 10] & $\eta$: log [0.0001, 0.01]; clip: log [0.1, 10] \\
VRAdam & $\eta$: log [0.0001, 0.3]; $\beta_3$: uni [0.1, 5]; $\alpha_1$: int [3, 30] & $\eta$: log [0.0001, 0.01]; $\beta_3$: log [0.001, 10]; $\alpha_1$: int [3, 30] \\
K-Adam quartic & $\eta$: log [0.0001, 0.3]; $c$: log [0.05, 2] & $\eta$: log [0.0001, 0.01]; $c$: log [0.05, 2] \\
K-Adam quartic (EMA $c$) & $\eta$: log [0.0001, 0.3]; $c$: log [0.5, 4] & $\eta$: log [0.0001, 0.01]; $c$: log [0.5, 4] \\
K-Adam relativistic & $\eta$: log [0.0001, 0.3]; $c$: log [0.05, 2] & $\eta$: log [0.0001, 0.01]; $c$: log [0.05, 2] \\
K-Adam VR map (abl.) & $\eta$: log [0.0001, 0.3]; $c$: log [0.05, 2] & $\eta$: log [0.0001, 0.01]; $c$: log [0.05, 2] \\
K-Adam quartic, global (abl.) & $\eta$: log [0.0001, 0.3]; $c$: log [0.05, 2] & $\eta$: log [0.0001, 0.01]; $c$: log [0.05, 2] \\
EoS controller (precond.) & $\eta$: log [0.0001, 0.3]; $\sigma$: uni [0.3, 2] & $\eta$: log [0.0001, 0.01]; $\sigma$: uni [0.3, 2] \\
EoS controller (directional) & $\eta$: log [0.0001, 0.3]; $\sigma$: uni [0.3, 2] & $\eta$: log [0.0001, 0.01]; $\sigma$: uni [0.3, 2] \\
Muon & $\eta$: log [0.001, 0.3]; $\eta_{\rm Adam}$: log [0.0001, 0.03] & $\eta$: log [0.003, 0.1]; $\eta_{\rm Adam}$: log [0.0001, 0.01] \\
Muon + quartic map & $\eta$: log [0.001, 0.3]; $\eta_{\rm Adam}$: log [0.0001, 0.03]; $c$: log [0.5, 4] & $\eta$: log [0.003, 0.1]; $\eta_{\rm Adam}$: log [0.0001, 0.01]; $c$: log [0.5, 4] \\
\bottomrule
\end{tabular}
}
\end{center}
Basin experiments (\cref{sec:global}): \BasinNZero{} initial conditions, 15{,}000 iterations (40{,}000 for $\beta_1\ge0.95$; runs with one tenth of the iterations give identical shares to two decimals); an orbit counts as reached when $|y_{t+1}+y_t|<10^{-6}(1+|y_t|)$ over the last four iterates. The Legendre speed is computed by Newton's method from the upper bound $\min(z,z^{1/n})$, which converges monotonically (residual $<2\times10^{-15}$ relative). Large starts: $|y_0|$ (in units of $c$) on a log grid of 201 points in $[1,10^5]$, $3\times10^4$ iterations ($8\times10^4$ for $\beta_1=0.99$). In the exact quadratic a run counts as on the orbit when the alternation index is at least $0.99$, $|\Reff-1|\le0.01$ and the amplitude is within $10^{-3}$ of the closed form. Exact quadratic: integration in the eigenbasis of $C$, which leaves the tensor norm and hence the map unchanged (checked against direct integration, relative difference $10^{-12}$). Delay study (\cref{sec:delay}): $R$ swept from $0.5$ to $4$, initial deviation $10^{-3}c$, 64 noise realizations per rate; peak $\Reff$ is the maximum of $R\,g$ along the sweep; the linear estimate is the $R$ at which the propagated amplitude reaches half the orbit amplitude. Sensitivity at the ResNet rate: sweep to $R=5$, initial deviations $10^{-6}$--$1$ (in units of $c$), cubic and quartic maps. Controlled ramps (\cref{fig:ramp}): float64 throughout; pretraining 1{,}500 steps at $R=0.5$ with $\eta$ reset from the top eigenvalue every 50 steps; power iteration warm-started from the previous eigenvector (200 iterations at the start of a ramp, then four every five steps); the two copies share the momentum at the kick; a run stops ten steps after $|x_t|$ first exceeds 100 times the exit threshold; at $\rho=0.001$ the kick is $\delta=10^{-4}$ so that the decayed deviation stays above the rounding level of the trajectory difference; noise is added to the gradient of the second copy only; the model distributions use 400 realizations on the measured path of each run. GPU studies: curvature computations in FP32 with TF32 disabled and deterministic cuDNN; training in FP32 for the full-batch studies and the grokking test, and with bf16 autocast for the minibatch ResNet-32, ResNet-18 and GPT studies (\cref{tab:protocol}); Lanczos with full re-orthogonalization, eigenvalue-error and stagnation stopping, two independent starts for calibration.

\end{document}